\documentclass[12pt]{article}

\usepackage{latexsym,amsfonts,amsmath,
epsfig,
tabularx,amsthm,dsfont,mathrsfs}

\usepackage{xcolor}
\usepackage{graphicx}

\usepackage{enumitem}

\usepackage{tikz}
\newcommand\copyrighttext{%
  \footnotesize Link to formal publication: https://doi.org/10.1016/j.spa.2019.06.015 \\
  \textcopyright 2019. This manuscript version is made available under the CC-BY-NC-ND 4.0 license http://creativecommons.org/licenses/by-nc-nd/4.0/ }
\newcommand\copyrightnotice{%
\begin{tikzpicture}[remember picture,overlay]
\node[anchor=south,yshift=10pt] at (current page.south) {\fbox{\parbox{\dimexpr\textwidth-\fboxsep-\fboxrule\relax}{\copyrighttext}}};
\end{tikzpicture}%
}

\theoremstyle{definition}

\newtheorem{thm}{Theorem} 
\newtheorem{rem}[thm]{Remark}
\newtheorem{alg}{Algorithm}
\newtheorem{lem}[thm]{Lemma}
\newtheorem{prop}[thm]{Proposition}
\newtheorem{cor}[thm]{Corollary}

\newtheorem{ass}[thm]{Assumption}

\newcommand\N{\mathbb{N}}

\newcommand{\wt}{\widetilde}
\newcommand{\widebar}[1]{\mbox{\kern1.5pt\hbox{\vbox{\hrule height 0.6pt \kern0.35ex
        \hbox{\kern-0.15em \ensuremath{#1 }\kern0.0em}}}}\kern-0.1pt}

\newcommand{\E}{\mathbb{E}}

\newcommand{\dint}{{\rm d}}

\newcommand{\abs}[1]{\left\vert #1 \right\vert}
\newcommand{\norm}[1]{\left\Vert #1 \right\Vert}

\newlength{\fixboxwidth}
\title{\vspace{-50pt}
Perturbation Bounds for\\ Monte Carlo within Metropolis\\ via Restricted Approximations
} 

\date{}

\author{
Felipe Medina-Aguayo\thanks{Department of Mathematics and Statistics, University of Reading Whiteknights, PO Box 220, Reading RG6 6AX, United Kingdom, email: f.j.medinaaguayo@reading.ac.uk} 
\and  
Daniel Rudolf\thanks{Institute for Mathematical Stochastics,  Universit\"at G\"ottingen \& Felix-Bernstein-Institute for Mathematical Statistics, Goldschmidtstra\ss e 7, 37077 G\"ottingen,  Germany,  email: daniel-rudolf@uni-goettingen.de} 
\and  
Nikolaus Schweizer\thanks{Department of Econometrics and OR, Tilburg University, PO Box 90153, 5000 LE Tilburg, The Netherlands, email: n.f.f.schweizer@uvt.nl}
}

\begin{document}
\maketitle
\copyrightnotice

\begin{abstract}
\noindent
The Monte Carlo within Metropolis (MCwM) algorithm, interpreted as a perturbed Metropolis-Hastings (MH) algorithm, provides an approach for approximate sampling when the target distribution is intractable. Assuming the unperturbed Markov chain is geometrically ergodic, we show explicit estimates of the difference between the $n$-th step distributions of the perturbed MCwM and the unperturbed MH chains. These bounds are based on novel perturbation results for Markov chains which are of interest beyond the MCwM setting. To apply the bounds, we need to control the difference between the transition probabilities of the two chains and to verify stability of the perturbed chain.
\end{abstract}

\section{Introduction}
The \emph{Metropolis--Hastings} (MH) algorithm is a classical method for sampling approximately 
from a distribution of interest
relying only on point-wise evaluations of an \emph{unnormalized} 
density. However, when even this 
unnormalized
density depends on unknown integrals and cannot easily
be evaluated, then this approach is not feasible. A possible solution
is to replace the required density evaluations in the MH acceptance
ratio with suitable approximations. This idea is implemented in \emph{Monte Carlo within Metropolis}
(MCwM) algorithms which substitute the unnormalized density evaluations by Monte Carlo estimates 
 for the intractable integrals.

Yet in general, replacing the exact MH acceptance ratio by an approximation
leads to inexact algorithms in the sense that a stationary distribution of the transition kernel of the resulting 
	Markov chain (if it exists) is not the distribution of interest.
	Moreover, convergence to a distribution is not at all clear.
Nonetheless, these approximate, perturbed,
or noisy methods, see e.g. \cite{AFEB14,JoMa17B,Joetal15}, have recently
gained increased attention due to their applicability in certain intractable
sampling problems. 
In this work we attempt to answer the following questions
about the MCwM algorithm:
\begin{itemize}
\item Can one quantify the quality of MCwM algorithms?
\item When might the MCwM algorithm fail and what can one do in such situations?
\end{itemize}

Regarding the first question, 
by using bounds on the difference of the $n$-th step distributions of {a MH and a MCwM algorithm based Markov chain}
we give a positive answer. 
For the second question, we suggest a modification for stabilizing the MCwM approach by restricting
the Markov chain to a suitably chosen set that contains the ``essential part'', 
{which we also call the ``center''}
of the state space. 
We provide examples where this restricted version of MCwM converges towards the 
distribution of interest
while 
the unrestricted version does not. Note also that in practical implementations of Markov chain Monte Carlo on a computer,
simulated chains are effectively restricted to compact state spaces due to memory limitations. 
Our results on restricted approximations can also be read in this spirit. 

\noindent\textbf{{Perturbation theory}.}
Our overall approach is based on {perturbation theory} for Markov chains. Let $(X_n)_{n\in \N_0}$ be a Markov chain  
with transition kernel $P$ and $(\widetilde X_n)_{n\in\mathbb{N}_0}$ be a Markov chain
with transition kernel $\widetilde P$ on a common Polish space $(G,\mathcal{B}(G))$. 
We think of $P$ and $\widetilde P$ as ``close'' to each other in a suitable sense and consider $\widetilde P$ as a perturbation of $P$.
In order to quantify the difference of the distributions of $X_n$ and $\widetilde X_n$,
denoted by $p_n$ and $\widetilde p_n$ respectively, we work with
\begin{equation}\label{eq:norm.pns}
\norm{ p_n - \wt p_n }_{\rm tv},
\end{equation}
where $\norm{\cdot}_{\rm tv}$ 
denotes the total variation 
distance. 
The Markov chain $(X_n)_{n\in\mathbb{N}_0}$ 
can be interpreted as the unavailable, unperturbed, or
ideal chain; while $(\wt X_n)_{n\in \mathbb{N}_0}$ is a perturbation that is available for simulation. 
We focus on the case where the ideal Markov chain is \emph{geometrically ergodic}, more
precisely \emph{$V$-uniformly ergodic}, {implying that its transition kernel $P$} satisfies a \emph{Lyapunov condition} of the form
\begin{align*}
P V(x) \leq \delta  V(x) + L, \qquad x\in G,  
\end{align*}
for some function $V\colon G\to [1,\infty)$ and numbers 
$\delta\in [0,1), L\in [1,\infty)$. 

To obtain estimates of  \eqref{eq:norm.pns} we need two 
assumptions which can be informally 
explained
as follows:
\begin{enumerate}
 \item \emph{Closeness of $\widetilde P$ and $P$:}
{
The difference of $\widetilde P$ and $P$ is measured by controlling either
	a weighted total variation distance or a weighted 
	$V$-norm of $P(x,\cdot)-\widetilde P(x,\cdot)$ uniformly. Here, uniformity either refers to the entire state space or, at least, to the ``essential'' part of it.}
 \item \emph{Stability of $\widetilde P$:} {
 A stability condition on $\widetilde P$ is satisfied either in the form of a Lyapunov condition  or by restriction to the center of the state space determined by $V$.} 

\end{enumerate}
{Under these assumptions, explicit bounds on \eqref{eq:norm.pns} are provided in Section~\ref{sec: quant_est}. More precisely, in Proposition~\ref{thm: est1} and Theorem~\ref{thm: est2} stability is guaranteed through a Lyapunov condition for $\widetilde P$, whereas in Theorem~\ref{thm: restricted_approx} a \emph{restricted approximation} $\wt P$ is considered.}

\noindent\textbf{{Monte Carlo within Metropolis}. }
 {In Section~\ref{sec: approx_MH}
we apply our perturbation bounds in the context of approximate sampling via 
MCwM. In the following we briefly introduce the setting.} 
The goal is to (approximately) sample from a target distribution $\pi$ 
on $G$, which is determined by an unnormalized density function $\pi_u\colon G\to [0,\infty)$ w.r.t
a reference measure $\mu$, that is,
\[
 \pi(A) = \frac{\int_A \pi_u(x) \,\dint\mu(x)}{\int_G \pi_u(x) \,\dint \mu(x)}, \quad A\in\mathcal{B}(G).
\]
Classically
the method of choice is to construct a Markov chain $(X_n)_{n\in\mathbb{N}_0}$
based on the MH algorithm for approximate sampling of $\pi$. This algorithm crucially relies
on knowing (at least) the ratio $\pi_u(y)/\pi_u(x)$ 
for arbitrary $(x,y)\in G^2$, e.g., because $\pi_u(x)$ and $\pi_u(y)$ can readily be computed.
However, in some scenarios, only approximations of $\pi_u(x)$ and $\pi_u(y)$ are available. 
Replacing the true unnormalized density $\pi_u$ in the MH algorithm by an approximation yields a perturbed, 
``inexact'' Markov chain $(\wt X_n)_{n\in\mathbb{N}_0}$. 
If the approximation is based on a Monte Carlo method, 
the perturbed chain is called MCwM chain.

Two particular settings where approximations of $\pi_u$ may 
rely on Monte Carlo estimates are \emph{doubly-intractable distributions} and \emph{latent variables}.
Examples of the former occur in Markov or Gibbs random fields, where the {function values $\pi_u(x)$} of the
unnormalized density itself are only known up to a factor $Z(x)$. This means that 
\begin{equation}  \label{eq: doubly_intractable_INTRO}
 \pi_u(x) = \rho(x)/ Z(x), \qquad x\in G,
 \end{equation}
where only values of $\rho(x)$ can easily be computed while the computational problem lies in evaluating
\[
Z(x) = \int_\mathcal{Y} \widebar{\rho}(x,y) r_x(\dint y),
\]
where $\mathcal{Y}$ denotes an auxiliary variable space, $\widebar{\rho}\colon G\times \mathcal{Y} \to [0,\infty)$ 
and $r_x$ {is a 
	probability distribution on $\mathcal{Y}$}.
We {investigate a} MCwM algorithm, which in every transition 
uses an iid sequence of random variables $(Y_i^{(x)})_{1\leq i\leq N}$,
with $Y_1^{(x)}\sim r_x$, to approximate $Z(x)$ by 
$ \widehat Z_N(x) := \frac{1}{N} \sum_{i=1}^N \widebar{\rho}(x,Y^{(x)}_i)$ 
(and $Z(y)$ by $\widehat Z_N(y)$, respectively). 
The second setting we study arises from \emph{latent variables}. Here, $\pi_u(x)$ cannot be evaluated since it takes the form
\begin{equation} 
\pi_u(x) = \int_{\mathcal{Y}} \widebar{\rho}(x,y)\, r_x(\dint y),
\end{equation}
where $r_x$ is a probability distribution 
on a measurable space $\mathcal{Y}$ of 
latent variables $y$, and $\widebar \rho \colon G\times \mathcal{Y} \to [0,\infty)$ 
is a non-negative density function. 
In general, no explicit computable expression 
of the above integral is at hand and the MCwM 
idea is to substitute $\pi_u(x)$ in the MH algorithm
by a Monte Carlo estimate based on 
iid sequences of random variables $(Y^{(x)}_i)_{1\leq i\leq N}$ 
and $(Y^{(y)}_i)_{1\leq i \leq N}$ with $Y_1^{(x)}\sim r_x$, $Y_1^{(y)}\sim r_y$. 
The resulting MCwM algorithm has been studied 
before in \cite{AnRo09,MeLeRo15}. 
{Let us note here that
this
MCwM approach should not be confused 
with the pseudo-marginal method, see \cite{AnRo09}.}
The pseudo-marginal method constructs a Markov chain on the extended space  $G\times \mathcal{Y}$
that targets a distribution with $\pi$ as its marginal on $G$.

\noindent\textbf{{Perturbation bounds for MCwM.}}
In both intractability settings,
the corresponding MCwM  Markov chains depend 
on the parameter $N\in \N$ which denotes the number of samples used within the Monte Carlo
estimates. As a consequence, any bound on 
\eqref{eq:norm.pns} is $N$-dependent, which allows us to control 
the dissimilarity to the ideal MH based Markov chain. 
In Corollary~\ref{cor: infbounds_doubly_intr} and the application of Corollary~\ref{cor: doub_restr_MCwM}
to the examples considered in Section~\ref{sec: approx_MH}
we provide
informative rates of convergence as $N \rightarrow \infty$.
Note that with those estimates we relax the requirement of uniform bounds on the approximation error
introduced by the estimator for $\pi_u$, which is essentially
imposed in \cite{MeLeRo15,AFEB14}.
In contrast to this requirement, we use (if available) the Lyapunov function as a counterweight 
for a second as well as inverse second moment
and can therefore handle situations where uniform bounds on the approximation error are not available.  
If we do not have access to a Lyapunov function for 
the MCwM transition kernel we suggest to restrict it to a subset of the state space, i.e., use restricted approximations.
This subset is determined by $V$ and usually corresponds to a ball with some radius {$R(N)$}
 that increases as the approximation quality improves, that is, $R(N)\rightarrow \infty$ as
$N\rightarrow \infty$.  

Our analysis of the MCwM algorithm is guided by some 
facts we observe in simple illustrations, 
in particular, 
we consider a log-normal example 
discussed in Section~\ref{sec: doubly_intractable}. 
In this example, 
we encounter a situation where the mean squared error of the Monte Carlo approximation 
grows exponentially in the tail of the target distribution. 
We observe \emph{empirically} that (unrestricted) MCwM works well 
whenever the growth behavior 
is dominated by the decay of the (Gaussian) target density in the tail. 
The application of Corollary~\ref{cor: doub_restr_MCwM} to the log-normal example shows
that the restricted approximation converges towards the true target density in the number of samples $N$ at least 
like $(\log N)^{-1}$ independent of \emph{any} growth of the error.
However, the convergence is better, at least like $\frac{\log N}{N}$, 
if the growth
is dominated by the decay of the target density. 

\section{Preliminaries}\label{sec: prelim}
Let $G$ be a Polish space, where $\mathcal{B}(G)$
denotes its Borel $\sigma$-algebra. Assume that $P$ is a transition kernel with stationary distribution $\pi$
on $G$. For a signed measure $q$ on $G$ 
and a measurable function
$f\colon G\to \mathbb{R}$ we define
\begin{align*}
 q P (A) & := \int_G P(x,A)\, \dint q(x),
 \quad 
 Pf(x) := \int_G f(y)\,P(x,\dint y),\quad x\in G, A\in \mathcal{B}(G).
\end{align*}
For a distribution $\mu$ on $G$ we use the notation
$
 \mu(f) := \int_G f(x)\,\dint \mu(x).
$
For a measurable function $V\colon G \to [1,\infty)$ 
and two probability measures $\mu,\nu$ on 
$G$
define
\[
 \norm{\mu-\nu}_{V} 
 := \sup_{\abs{f}\leq V} \abs{
\mu(f)-\nu(f)
 }.
\] 
For the constant function
$V=1$ this is the total variation distance, i.e.,
\[
 \norm{\mu-\nu}_{\rm tv} 
 := \sup_{\abs{f}\leq 1} 
\abs{
\mu(f)-\nu(f)
}.
\]
The next, well-known theorem defines geometric ergodicity
and states a useful equivalent condition. The proof follows by  \cite[Proposition~2.1]{RoRo97}
and \cite[Theorem~16.0.1]{MeTw09}.
\begin{thm}  \label{thm: equi}
 For a $\phi$-irreducible and aperiodic transition kernel $P$
 with stationary distribution $\pi$ defined on $G$
 the following statements are equivalent:
 \begin{itemize}
  \item The transition kernel $P$ is \emph{geometrically ergodic}, that is, 
 there exists
 a number $\bar{\alpha}\in [0,1)$ and 
 a measurable function
$C\colon G\to [1,\infty)$ such that for $\pi$-a.e.
$x\in G$ we have
\begin{equation}
 \norm{P^n(x,\cdot)-\pi}_{\rm tv} \leq C(x) \bar{\alpha}^n, \quad n\in \mathbb{N}.
\end{equation}
  \item There is a $\pi$-a.e. finite measurable function $V\colon G\to [1,\infty]$
 with finite moments with respect to $\pi$ and 
 there 
 are constants $\alpha\in[0,1)$ 
 and 
 $C\in[1,\infty)$ such that 
 \begin{equation}  \label{eq: V-unif}
  \norm{P^n(x,\cdot)-\pi}_{V}
  \leq C V(x) \alpha^n,\quad x\in G,\; n\in \mathbb{N}.
 \end{equation} 
 \end{itemize}
 In particular, the function $V$ can be chosen such that a \emph{Lyapunov condition}
 of the form
 \begin{equation}  \label{eq: def_Lyap}
    P V(x) \leq \delta  V(x) + L, \qquad x\in G,  
 \end{equation}
 for some $\delta\in [0,1)$ and $L\in (0,\infty)$, is satisfied.
\end{thm}

\begin{rem}
We call a transition kernel \emph{$V$-uniformly ergodic}
if it satisfies \eqref{eq: V-unif} and note that this condition 
can be
be rewritten as
\begin{equation}  \label{eq: V_unif_erg_ratio}
 \sup_{x\in G}\frac{\norm{P^n(x,\cdot)-\pi}_{V}}{V(x)} \leq C \alpha^n.
\end{equation}
\end{rem}

\section{Quantitative perturbation bounds}
\label{sec: quant_est}

Assume that $(X_n)_{n\in \mathbb{N}_0}$ is a Markov chain with transition kernel
$P$ and initial distribution $p_0$ on $G$.
We define
$p_n:= p_0 P^n$, i.e., $p_n$ is the distribution of $X_n$. 
The distribution $p_n$ is approximated by using another Markov chain 
$(\wt X_n)_{n\in \mathbb{N}_0}$ with transition kernel
$\wt P$ and initial distribution $\wt p_0$.
We define
$\wt p_n := \wt p_0 \wt P^n$, i.e., $\wt p_n$ is the distribution of $\wt X_n$.
The idea throughout the paper is to interpret $(X_n)_{n\in \mathbb{N}_0}$
as some ideal, unperturbed chain and $(\wt X_n)_{n\in \mathbb{N}_0}$
as an approximating, perturbed Markov chain. 

In the spirit of the doubly-intractable distribution
and latent variable case considered in Section~\ref{sec: approx_MH} 
we think of the unperturbed Markov chain as ``nice'', where convergence properties
are readily available.
Unfortunately since we cannot simulate the ``nice'' chain we try to approximate it with 
a perturbed Markov chain, which is, because of the perturbation, difficult to analyze directly.
With this in mind, we make the following 
standing assumption on the unperturbed Markov chain.

\begin{ass} \label{ass: erg}
  Let $V\colon G \to [1,\infty)$ be a measurable function and assume that $P$ is $V$-uniformly
 {ergodic, that is, \eqref{eq: V-unif} holds} for some constants $C\in [1,\infty)$ and $\alpha\in [0,1)$.
\end{ass}

We start with an auxiliary estimate of $\norm{p_n - \wt p_n}_{\rm tv}$ 
which is interesting on its own and is proved in \ref{sec: proof_sec3}.
\begin{lem} \label{lem: aux}
 Let Assumption~\ref{ass: erg} be satisfied and
 for a measurable function $W\colon G\to [1,\infty)$ define 
  \begin{align*}
    \varepsilon_{{\rm tv}, W}  &:= \sup_{x\in G} 
    \frac{\norm{P(x,\cdot)-\wt P(x,\cdot)}_{\rm tv}}{W(x)},\\    
   \varepsilon_{V, W} 
    & :=  \sup_{x\in G} \frac{\norm{P(x,\cdot)-\wt P(x,\cdot)}_{V}}{W(x)}.
 \end{align*} 
Then, for any $r\in (0,1]$,
\begin{equation}  \label{eq: aux_est}
 \norm{p_n-\wt p_n}_{\rm tv}
 \leq C\alpha^n \norm{p_0-\wt p_0}_V + \varepsilon_{{\rm tv},W}^{1-r}\; \varepsilon_{V,W}^r\;	 C^r
 \sum_{i=0}^{n-1} \wt p_i(W) \alpha^{(n-i-1)r}.
\end{equation}
\end{lem}

\begin{rem}
 The quantities $\varepsilon_{{\rm tv},W}$ and $\varepsilon_{V,W}$ 
 measure the difference between
 $P$ and $\wt P$. Note that we can interpret them as operator norms 
 \begin{align*}
  \varepsilon_{{\rm tv},W} & = \norm{P-\wt P}_{B^{(1)} \to B^{(W)}}
  \quad
  \text{and}
  \quad
  \varepsilon_{V,W} = \norm{P-\wt P}_{B^{(V)} \to B^{(W)}},
 \end{align*}
 where
  \begin{equation}  \label{eq: fct_class_norm}
     B^{(W)} = \left\{ f\colon G \to \mathbb{R}
  \mid \norm{f}_{\infty,W} := \sup_{x\in G}\frac{\abs{f(x)}}{W(x)} < \infty \right\}.
  \end{equation}
 {It is also easily seen that $\varepsilon_{{\rm tv},W} \leq \min\{2,\varepsilon_{V,W}\}$
 which implies that 
 a small number $\varepsilon_{V,W}$ leads also to a small number $\varepsilon_{{\rm tv},W}$.}
 In \eqref{eq: aux_est} an additional parameter $r$ appears which can be
 used to tune the estimate. Namely, if one is not able to bound $\varepsilon_{V,W}$
 sufficiently well but has a good estimate of $\varepsilon_{{\rm tv},W}$ 
 one can optimize over $r$.
 On the other hand, if there is a satisfying estimate of $\varepsilon_{V,W}$ 
 one can just set $r=1$.
\end{rem}

In the previous lemma we proved an upper bound of $\norm{p_n-\wt p_n}_{\rm tv}$ 
which still contains an unknown quantity given by
\[
 \sum_{i=0}^{n-1} \wt p_i(W) \alpha^{(n-i-1)r}
\]
which measures, in a sense, stability of the perturbed chain through a weighted sum of expectations of the Lyapunov function $W$ under $\wt p_i$. To control this term, 
we impose additional assumptions on the perturbed chain. In the following, we consider two assumptions of this type, a Lyapunov condition and a bounded support assumption. 

\subsection{Lyapunov condition}
We start with a simple version of our 
main estimate which illustrates already some key aspects of
the approach via the Lyapunov condition. 
Here the intuition is as follows: By Theorem~\ref{thm: equi} we know
that the function $V$ of Assumption~\ref{ass: erg}
can be chosen such that a Lyapunov condition for $P$ is satisfied.
Since we think of $\wt P$ as being close to $P$, it might be possible to show also a Lyapunov
condition with $V$ of $\wt P$. If this is the case, the following proposition is applicable.
\begin{prop}   \label{thm: est1}
 Let Assumption~\ref{ass: erg} be satisfied.
 Additionally, {let $\widetilde{\delta}\in[0,1)$ and $\widetilde{L}\in (0,\infty)$ be such
 that
 \begin{align}  
  \label{eq: Lyap}
  \wt P V(x) & \leq \wt \delta \,V(x) + \wt L,\qquad x\in G.
 \end{align}
Assume that $p_0 = \wt p_0$
and define
$
 \kappa := \max\left\{
\wt p_0(V)
 , \frac{\wt L}{1-\wt \delta} \right\},
$}
as well as 
(for simplicity)
 \begin{align*}
    \varepsilon_{\rm tv} := \varepsilon_{{\rm tv},V},
\qquad
\varepsilon_V := \varepsilon_{V,V}. 
\end{align*}
 Then, for any $r\in(0,1]$, 
\begin{equation}  \label{eq: tv_est_simple}
 \norm{p_n-\wt p_n}_{\text{tv}}
\leq 
 \varepsilon_{\text{tv}}^{1-r} \varepsilon_V^{r}\;   \frac{C^r \kappa}{(1-\alpha)r}.
\end{equation}
\end{prop}
\begin{proof}
We use Lemma~\ref{lem: aux} with $W=V$.
{By \eqref{eq: Lyap}, it
 follows
that
 \begin{equation} 
\wt p_i(V)
  = \int_G \wt P^i {V}(x)\, \wt p_0(\dint x)
  \leq \wt \delta^i \wt p_0( V)
  + (1-\wt \delta^i) \frac{\wt L}{1-\wt \delta} \leq \kappa.
 \end{equation}
 The} final estimate is obtained by a geometric series and $1-\alpha^r \geq r(1-\alpha)$.
\end{proof}

Now we state a more general theorem.
In particular, in this estimate the dependence 
on the initial distribution
can be weakened. In the perturbation bound of the previous estimate, 
the initial distribution is only forgotten if 
{$\wt p_0(V) < \wt L/(1-\wt\delta)$}. Yet, intuitively, for long-term stability results
$\wt p_0(V)$ should not matter at all. This intuition is confirmed by the theorem. 

\begin{thm}  \label{thm: est2}
Let Assumption~\ref{ass: erg} be satisfied.
 Assume also that $W\colon G \to [1,\infty)$ is a measurable function which satisfies 
 {with $\wt\delta\in[0,1)$ and $\wt L\in (0,\infty)$
 the Lyapunov condition 
 \begin{align}  \label{eq: Lyup}
  \wt P W (x) & \leq \wt\delta W(x) + \wt L,\qquad x\in G.
 \end{align}
Define 
$\varepsilon_{{\rm tv},W}$, $\varepsilon_{V,W}$ as in Lemma~\ref{lem: aux}
and 
 $\gamma  := \frac{\wt L}{1-\wt\delta}$.
 Then, for any $r\in(0,1]$ with 
 \begin{equation*}
  \beta_{n,r}(\wt\delta,\alpha) := 
  \begin{cases}
     n \alpha^{(n-1)r}, & \alpha^r = \wt\delta,\\
     \frac{\abs{\alpha^{rn} - \wt\delta^n}}{\abs{\alpha^r - \delta}},& \alpha^r \not =\wt\delta,
  \end{cases}
 \end{equation*}
 we have
 \begin{equation}  \label{eq: fin_est}
  \norm{\wt p_n- p_n}_{\rm tv} 
  \leq C  \alpha^n \norm{\wt p_0 - p_0}_{V} 
  +\varepsilon_{{\rm tv},W}^{1-r}\;\varepsilon_{V,W}^{r}\; C^r
    \left[ \wt p_0(W)\, 
\beta_{n,r}(\wt\delta,\alpha)
    + \frac{\gamma}{(1-\alpha)r}  \right].
 \end{equation}}
\end{thm}

\begin{proof}
Here we use Lemma~\ref{lem: aux} with possibly different $W$ and $V$.
{By \eqref{eq: Lyup}
we have
$
\wt p_i(W)
  \leq \wt\delta^i \wt p_0( W)
  + 
  \gamma
$
and by
\[
 \sum_{i=0}^{n-1} \wt\delta^i \alpha^{(n-i-1)r} 
 = \beta_{n,r}(\wt\delta,\alpha)
\]
we obtain} the assertion by a geometric series and $1-\alpha^r \geq r(1-\alpha)$.
\end{proof}

\begin{rem}
\label{rem: beta}
We consider an illustrating example where Theorem~\ref{thm: est2} leads to a considerably sharper bound than 
Proposition~\ref{thm: est1}. This improvement is due to the combination of two novel properties of the bound of Theorem~\ref{thm: est2}:
\begin{enumerate}
 \item In the Lyapunov condition \eqref{eq: Lyup} the function
 $W$ can be chosen differently from $V$.
 \item {Note that
 $\beta_{n,r}(\wt\delta,\alpha)$ is bounded from above by $n\cdot \max\{\wt\delta,\alpha^r\}^{n-1}$. Thus $\beta_{n,r}(\wt\delta,\alpha) $
 converges almost exponentially} fast to zero in $n$. This implies that 
 for $n$ sufficiently large the dependence of $\wt p_0(W)$ vanishes. Nevertheless, the leading factor $n$ can capture situations
in which the perturbation error is increasing in $n$ for small $n$. 
\end{enumerate}
\end{rem}

\noindent
{\bf Illustrating example.}
Let $G=\{0,1\}$ and assume $p_0=\wt p_0 = (0,1)$. Here state ``$1$'' can be interpreted 
{as ``transitional'' while state ``$0$'' as ``essential'' part of the state space.} 
Define
\[
P= \left(\begin{array}{cc} 
1 & 0\\ 1 & 0
\end{array}\right)  
\qquad
\text{and}
\qquad
\wt P= \left(\begin{array}{cc} 
1 & 0\\ \frac12 & \frac12
\end{array}\right)  .
\]
Thus, the unperturbed Markov chain $(X_n)_{n\in \mathbb{N}_0}$ 
moves from ``$1$'' to ``$0$'' right away, 
while the perturbed one $(\wt X_n)_{n\in \mathbb{N}_0}$
takes longer. 
Both transition matrices have the same
stationary distribution $\pi=(1,0)$. 
Obviously, $\|p_0 - \wt p_0 \|_{\rm tv}=0$ and for $n\in \mathbb{N}$ it holds that
\[
\|p_n - \wt p_n \|_{\rm tv} =2 \mathbb{P}(X_n \neq \wt X_n) =\frac{1}{2^{n-1}}.
\]
The unperturbed Markov chain is uniformly ergodic, such that
we can choose $V=1$ and \eqref{eq: V-unif} is satisfied with $C=1$ and $\alpha=0$.
In particular, in this setting $\varepsilon_{\rm tv}$ and $\varepsilon_V$
from Proposition~\ref{thm: est1} coincide, we have $\varepsilon_{\rm tv}=1$.
Thus, the estimate of Proposition~\ref{thm: est1} gives
\[
\|p_n - \wt p_n \|_{\rm tv} \leq \varepsilon_{\rm tv} = 1.
\]
This bound is optimal in the sense that 
it is best possible for $n = 1$. 
But for increasing $n$ it is getting worse. 
Notice also that a different choice of $V$ cannot really remedy this situation: 
The chains differ most strongly at $n=1$ and 
 the bound of Proposition~\ref{thm: est1} is constant over time.  
Now choose the function 
$W(x) = 1 + v \cdot \mathbf{1}_{\{x=1\}} $ 
for some $v \geq 0$. 
The transition matrix
$\wt P$ satisfies the Lyapunov condition 
\[
\wt P W(x) \leq \frac{1}{2}\, W(x) + \frac{1}{2}, 
\]
i.e., {$\wt\delta = \wt L = \frac{1}{2}$}. 
Moreover, we have $\wt p_0(W)=1+v$ and 
$\varepsilon_{V,W} = \varepsilon_{{\rm tv},W} = 1/(1+ v)$. 
Thus, in the bound from Theorem~\ref{thm: est2} 
we can set $r=1$ and $\gamma = 1$ such that 
\[
\|p_n - \wt p_n \|_{\rm tv} \leq  \frac1{ v+1 } + \frac{1}{2^{n-1}}.
\]
Since $v$ can be chosen arbitrarily large, it follows that
\[
\|p_n - \wt p_n \|_{\rm tv} \leq  \frac{1}{2^{n-1}},
\]
which is best possible for all $n\in \mathbb{N}$.

The previous example can be seen as a toy model of 
a situation where the transition probabilities of a perturbed
and unperturbed Markov chain are very similar in the {``essential'' part of the state space, 
but differ considerably in the ``tail'', seen as the ``transitional'' part.} 
When the chains
start both at the same point in the ``tail'', 
considerable differences between distributions can
build up along the initial transient and then vanish again. 
Earlier perturbation bounds as for example in \cite{Mi05, PS14, RuSc15} 
take only an initial error and a remaining error into account. Thus, those 
are worse for situations where this transient error captured by $\beta_{n,r}$ dominates. 
A very similar term also appears in the very recent error bounds due to \cite{JoMa17B}.
{In any case, the example also illustrates that}
a function
$W$ different from $V$ is advantageous.

\subsection{Restricted approximation}
\label{sec: restr_approx}

In the previous section, we have seen that a Lyapunov condition 
of the perturbation helps to control 
the long-term stability of approximating
a $V$-uniformly ergodic Markov chain.
In this section we assume that the perturbed chain 
is restricted to a ``large'' subset
of the state space. In this setting a sufficiently good approximation of the unperturbed Markov chain
on this subset leads 
to a perturbation estimate.

For the unperturbed Markov chain we assume that transition
kernel $P$ is $V$-uniformly ergodic. 	
Then, for $R\geq 1$ define the ``large subset'' of the state space as 
\[
B_R=\{x \in G \mid V(x) \leq R \}.
\]
If $V$ is chosen as a monotonic transformation of a norm on $G$, $B_R$
is simply a ball around $0$. 
The \emph{restriction of $P$} to the set $B_R$, given as $P_R$, is defined as
\[
 P_R(x,A) = P(x,A\cap B_R) + \mathbf{1}_A(x) P(x,B_R^c), \quad A\in \mathcal{B}(G),\,x\in G.
\]
In other words, whenever $P$ would make a transition from $x\in B_R$ to $G\setminus B_R$, 
$P_R$ remains in $x$. Otherwise, $P_R$ is the same as $P$. 
We obtain the following perturbation bound for approximations whose stability is guaranteed through a restriction to the 
set $B_R$.

\begin{thm}  \label{thm: restricted_approx}
 Under the $V$-uniform ergodicity of Assumption~\ref{ass: erg}
 let 
 $\delta \in [0,1)$ and $L\in [1,\infty)$ be chosen in such a way that
 \[
    PV(x) \leq \delta\,V(x) +L, \quad x\in G.
  \]
For the perturbed transition kernel $\wt P$ assume that it is restricted to $B_R$, i.e.,
$\wt P(x,B_R) = 1$ for all $x\in G$, and that $R\cdot\Delta(R)\leq (1-\delta)/2$ with
 \[
\Delta(R) := \sup_{x \in B_R} \frac{\norm{P_R(x,\cdot) - \wt{P}(x,\cdot)}_{\rm tv}}{V(x)}.
\]
Then, with $p_0=\wt p_0$ and
\[
\kappa:= \max\left \{ \wt p_0(V), \frac{L}{1-\delta} \right\}
\]
we have for $R\geq \exp(1)$ that
\begin{equation}
\label{eq: restr_approx_final_est}
  \norm{p_n-\wt p_n}_{\rm tv} \leq \frac{33C(L+1) \kappa}{1-\alpha} \cdot
  \frac{\log R}{R}.
\end{equation}
\end{thm}
\noindent
The proof of the result is stated in \ref{sec: proof_sec3}. 
Notice that while the perturbed chain is restricted to the set $B_R$, 
we do not place a similar restriction on the unperturbed chain. 
The estimate \eqref{eq: restr_approx_final_est} 
compares the restricted, perturbed chain to the unrestricted, unperturbed one. 
\begin{rem}
 In the special case where $\wt P(x,\cdot) = P_R(x,\cdot)$ for $x\in B_R$ we have $\Delta(R)=0$. 
 For example
 \[
  \wt P(x,A) = \mathbf{1}_{B_R}(x) P_R(x,A) + \mathbf{1}_{B_R^c}(x) \delta_{x_0}(A), \quad A\in \mathcal{B}(G),
 \]
 with $x_0\in B_R$ satisfies this condition. The resulting perturbed Markov chain
 is simply a restriction of the unperturbed Markov chain to $B_R$
 and Theorem~\ref{thm: restricted_approx} provides a quantitative bound on the difference 
 of the distributions.
\end{rem}
{
\subsection{Relationship to earlier perturbation bounds}
\label{sec: comments_lit}
In contrast to the $V$-uniform ergodicity assumption we impose
on the ideal Markov chain, the results in \cite{AFEB14,Joetal15,Mi05} 
only cover perturbations of uniformly ergodic Markov chains. 
Nonetheless, perturbation theoretical questions 
for geometrically ergodic Markov chains have 
been studied before, see e.g. \cite{BrRoRo01,FHL13,MeLeRo15,NeRo17,RoRoSch98,RuSc15,ShSt00} 
and the references therein. A crucial aspect where those papers differ from each other is how
one measures the closeness of the transitions of the unperturbed and perturbed Markov chains to
have applicable estimates, see the discussion about this in \cite{ShSt00,FHL13,RuSc15}.
Our Theorem~\ref{thm: est1} and Theorem~\ref{thm: est2} refine and extend the results of \cite[Theorem~3.2]{RuSc15}.
In particular, in Theorem~\ref{thm: est2} we take a restriction to 
the center of the state space into account.
Let us also mention here that \cite{PS14,RuSc15} 
contain related results under Wasserstein ergodicity assumptions.
More recently, \cite{JoMa17A} studies approximate chains using notions 
of maximal couplings, \cite{NeRo17} extends 
the uniformly ergodic setting from \cite{Joetal15} to 
using $L_2$ norms instead of total variation, 
and \cite{JoMa17B} explores bounds on the approximation error of time averages.} 

{
The usefulness of restricted approximations in the study of Markov chains has been observed before. } 
{For example in \cite{RuSp16}, 
	in an infinite-dimensional setting, spectral gap properties of a Markov operator based on a restricted approximation are investigated. 
Also recently in \cite{YaRo17} it is proposed to consider a subset of the state space termed ``large set'' in which a certain Lyapunov condition holds. This is in contrast to a Lyapunov function defined on the entire space, which might deteriorate as the dimension of the state space or the number of observations increases. This new Lyapunov condition from \cite{YaRo17} is particularly useful for obtaining explicit bounds on the number of iterations to get close to the stationary distribution in high-dimensional settings. 
}


\section{Monte Carlo within Metropolis}
\label{sec: approx_MH}

In Bayesian statistics it is of interest to sample
with respect to a distribution $\pi$ 
on $(G,\mathcal{B}(G))$.
We assume that $\pi$ admits a possibly \emph{unnormalized density} $\pi_u\colon G \to [0,\infty)$
with respect to a reference
measure $\mu$, for example the counting, Lebesgue or some Gaussian measure.
The Metropolis-Hastings (MH) algorithm is often the method of choice to
draw approximate samples according to $\pi$: 

\begin{alg} \label{alg: Metro_Hast}
For a \emph{proposal transition kernel $Q$} a transition from $x$
to $y$ of the MH algorithm works as follows.
\begin{enumerate}
 \item Draw $U\sim \text{Unif}[0,1]$ and a proposal $Z\sim Q(x,\cdot)$ independently, 
 call the result $u$ and $z$, respectively.
 \item Compute the {\emph{acceptance ratio}}
 \begin{equation}  \label{eq: acc_ratio}
  r(x,z):=\frac{\pi(\dint z)Q(z,\dint x)}{\pi(\dint x)Q(x,\dint z)} 
	 = \frac{\pi_u(z)}{\pi_u(x)}\frac{\mu(\dint z)Q(z,\dint x)}{\mu(\dint x)Q(x,\dint z)},
 \end{equation}
 which is the density of the measure $\pi(\dint z)Q(z,\dint x)$ w.r.t. $\pi(\dint x)Q(x,\dint z)$, see \cite{Ti98}.
 \item If {$u<r(x,z)$}, then accept the proposal, and return $y:=z$, 
 otherwise reject the proposal and return $y:=x$. 
\end{enumerate}
\end{alg}
The transition kernel of the MH algorithm with 
proposal $Q$, stationary distribution $\pi$ and acceptance probability 
 \[
  a(x,z) := \min\left\{ 1,  r(x,z) \right\}
 \]
is given by
\begin{equation}  \label{eq: MH_type}
 M_a(x,\dint z) := a(x,z) Q(x,\dint z) + \delta_x(\dint z) \left(1 - \int_{G} a(x,y) Q(x,\dint y)\right).
\end{equation}
For the MH algorithm in the computation of $r(x,z)$ one uses $\pi_u(z)/\pi_u(x)$, which
might be 
known from having access to 
function evaluations of the unnormalized density $\pi_u$.
However, when it is expensive or even impossible to compute function values of $\pi_u$, then it may not be feasible to sample from $\pi$ using the MH algorithm. 
Here are two typical examples of such scenarios:
\begin{itemize}
\item {\bf Doubly-intractable distribution:} 
For models such as \emph{Markov or Gibbs random fields}, the 
unnormalized density $\pi_u(x)$ itself is typically only known
 up to a factor $Z(x)$, that is, 
\begin{equation}  \label{eq: doubly_intractable}
 \pi_u(x) = \rho(x)/ Z(x), \qquad x\in G
 \end{equation}
where functions values of $\rho$ can be computed, but function values of $Z$ cannot. 
For instance, $Z$ might be given in the form
\[
 Z(x) = \int_\mathcal{Y} \widebar{\rho}(x,y)\, r_x(\dint y),
\]
where $\mathcal{Y}$ denotes an auxiliary variable space, 
$\widebar{\rho}\colon G\times \mathcal{Y} \to [0,\infty)$ and $r_x$ 
is a {probability distribution on $\mathcal{Y}$.}
 \item {\bf Latent variables:}  
 Here 
 $\pi_u(x)$ cannot be evaluated, since it takes the form
\begin{equation}  \label{eq: latent}
 \pi_u(x) = \int_{\mathcal{Y}} \widebar{\rho}(x,y)\, r_x(\dint y)
 \end{equation}
with a probability
 distribution $r_x$ on a 
measurable space $\mathcal{Y}$ of \emph{latent variables} $y$ and
a non-negative function $\widebar \rho \colon G\times \mathcal{Y} \to [0,\infty)$.
\end{itemize}

In the next sections, we study in both of these settings the perturbation error
of an approximating MH algorithm. A fair assumption in both scenarios, {which holds for a large family of target distributions using random-walk type proposals, see, e.g., \cite{MeTw96,RoTw96,JaHa00}}, is that the infeasible, unperturbed 
MH algorithm is $V$-uniformly ergodic:

\begin{ass}  \label{ass: metro_V_unif}
 For some function $V\colon G\to [1,\infty)$ 
 let the transition kernel $M_a$ of the MH algorithm be 
 $V$-uniformly ergodic, that is,
 \[
  \norm{M_a^n(x,\cdot) - \pi}_{V} \leq C V(x) \alpha^n
 \]
 with $C\in [1,\infty)$ and $\alpha\in [0,1)$, and additionally, assume that
 the Lyapunov condition
 \[
  M_aV(x) \leq \delta V(x) + L, 
 \]
 for some $\delta\in [0,1)$ and $L\in[1,\infty)$ is satisfied.
\end{ass}

We have the 
following standard proposition (see e.g. \cite[Lemma~4.1]{RuSc15} or \cite{AFEB14,BDH14,JoMa17B,MeLeRo17,PS14}) 
which leads to upper bounds on 
$\varepsilon_{{\rm tv}}$, $\varepsilon_{V}$ and $\Delta(R)$ (see Lemma~\ref{lem: aux} 
and Theorem~\ref{thm: restricted_approx})
for two MH type algorithms $M_b$ and $M_c$
with {common proposal distribution but different acceptance probability functions $b,c\colon G\times G\to [0,1]$, respectively}.

\begin{prop}
\label{prop: accept_err}
 Let $b,c \colon G\times G\to [0,1]$ and let $V\colon G\to [1,\infty)$ 
 be such that  ${\sup\limits_{x\in G}}\frac{M_b V(x)}{V(x)} \leq T$ for a constant $T\geq 1$.
 Assume that
 there are 
 functions $\eta, \xi \colon G\to [0,\infty)$ and a set $B\subseteq G$ such that, 
 {either
 \begin{equation} \label{eq: accept_err1case}
 \begin{aligned}
   \abs{b(x,y)-c(x,y)} &\leq \mathbf{1}_B(y)   (\eta(x)+\eta(y))b(x,y) \xi(x), \quad\text{or}\\
   \abs{b(x,y)-c(x,y)} &\leq \mathbf{1}_B(y)   (\eta(x)+\eta(y))b(x,y) \xi(y)
 \end{aligned}
 \end{equation}
 for all $x,y\in G$.
}
 Then we have
 \begin{equation*}
    \sup_{x\in B} \frac{\norm{M_b(x,\cdot)-M_c(x,\cdot)}_{V}}{V(x)}  \leq 4 T
  \norm{\eta\cdot \mathbf{1}_B}_{\infty}\norm{\xi \cdot \mathbf{1}_B}_{\infty},
 \end{equation*}
and, {with the definition of $\norm{\cdot }_{\infty,W}$ provided in \eqref{eq: fct_class_norm}}, for any $\beta\in (0,1)$, 
 \begin{align*}
  \sup_{x\in B} \frac{\norm{M_b(x,\cdot)-M_c(x,\cdot)}_{\rm tv}}{V(x)} & \leq 4 T 
  \norm{\eta\cdot \mathbf{1}_B}_{\infty,V^\beta}\norm{\xi\cdot \mathbf{1}_B}_{\infty,V^{1-\beta}}.
 \end{align*}
\end{prop}

The proposition provides 
a tool 
for controlling the distance between the transition kernels 
of two MH type algorithms with identical proposal and different acceptance 
probabilities. The specific functional form for the dependence of the upper bound in \eqref{eq: accept_err1case} on $x$ and $y$ is motivated by 
the applications below. 
The set $B$ indicates the ``essential'' part of $G$ where the 
difference of the acceptance probabilities matter.
The parameter $\beta$ is used 
to shift weight between the two components $\xi$ 
and $\eta$ of the approximation error. 
For the proof of the proposition, we refer to \ref{sec: proof_Sec4}.

\subsection{Doubly-intractable distributions}
\label{sec: doubly_intractable}

In the case where $\pi_u$ takes the form \eqref{eq: doubly_intractable}, we can 
approximate $Z(x)$ by a Monte Carlo estimate
\[
 \widehat Z_N(x) := \frac{1}{N} \sum_{i=1}^N \widebar{\rho}(x,Y^{(x)}_i),
\]
under the assumption that we have access to an iid 
sequence of random variables $(Y^{(x)}_i)_{1\leq i\leq N}$ 
where each $Y^{(x)}_i$ is distributed according to
$r_x$. Then, the idea is to substitute the unknown quantity $Z(x)$ by the approximation $\widehat Z_N(x) $ 
within the acceptance ratio. {Defining 
$W_N(x):= \frac{\widehat Z_N(x)}{Z(x)}$, 
the acceptance ratio can be written as} {
\begin{align*}
 \widetilde r(x,z,W_N(x),W_N(z)) &:= 
 \frac{\mu({\rm d} z)Q(z,{\rm d} x)}{\mu({\rm d} x)Q(x,{\rm d} z)} \cdot \frac{\widehat Z_N(x)}{\widehat Z_N(z)} =
 r(x,z) \cdot \frac{W_N(x)}{W_N(z)},
\end{align*} 
where} the random variables $W_N(x)$, $W_N(z)$ are assumed to be independent from each other. {Notice that the quantities $W_N$ only appear in the theoretical analysis of the algorithm.  For the implementation, it is sufficient to be able to compute $\wt r$.}
This leads to a \emph{Monte Carlo within Metropolis} (MCwM) algorithm:
\begin{alg} \label{alg: MCwM_doubly_intract}
For a given proposal transition kernel $Q$, a transition from $x$
to $y$ of the MCwM algorithm works as follows.
\begin{enumerate}
 \item Draw $U\sim \text{Unif}[0,1]$ and a proposal $Z\sim Q(x,\cdot)$ independently, 
 call the result $u$ and $z$, respectively.
 \item {Calculate $\widetilde r(x,z,W_N(x),W_N(z))$ based on independent samples for $W_N(x)$, $W_N(z)$, which are also independent from previous iterations.}
 \item If {$u<\widetilde  r(x,z,W_N(x),W_N(z))$}, then accept the proposal, and return $y:=z$, 
 otherwise reject the proposal and return $y:=x$. 
\end{enumerate}
\end{alg}
Given the current state $x\in G$ and a proposed state $z\in G$ the overall acceptance probability is{
\begin{align} \label{eq: doubl_intr_a_N}
a_N(x,z):=\mathbb{E}[ \min\left\{ 1, \widetilde r(x,z,W_N(x),W_N(z)) \right\} ],
\end{align}}
which leads to
the corresponding transition kernel of the form $M_{a_N}$, see \eqref{eq: MH_type}.
{
\begin{rem}
	Let us 
	emphasize that the doubly-intractable case can also be approached algorithmically from various other perspectives. For instance, instead of estimating the normalizing constant $Z(x)$ 
	one could estimate unbiasedly $(Z(x))^{-1}$ whenever exact simulation from the Markov or Gibbs 
	random field is possible.
	In this case, $\pi_u(x)$ turns into a Monte Carlo estimate which can formally be analyzed with exactly the same techniques as the latent variable scenario described below. 
	Yet another algorithmic possibility is explored in the \emph{noisy exchange} algorithm of  \cite{AFEB14}, where ratios of the form
	$Z(x)/Z(y)$ are approximated by a single Monte Carlo estimate. Their algorithm is motivated by 
	the \emph{exchange algorithm} \cite{MuGhMa06} which, 
	perhaps surprisingly, can avoid the need for evaluating the ratio $Z(x)/Z(y)$ and targets the distribution $\pi$ exactly, see e.g. \cite{EJRE16,PaHa18} for an overview of these and related methods.
	However, in some cases the exchange algorithm performs poorly, see \cite{AFEB14}. Then approximate 
	sampling
	methods for distributions of the form \eqref{eq: doubly_intractable_INTRO} might prove useful as long as the introduced bias is not too large. As a final remark in this direction, the recent work 
	\cite{ADYC_2018}  considers a correction of the noisy exchange algorithm which produces a 
	Markov chain with stationary distribution $\pi$.
\end{rem}
}
The quality of the MCwM algorithm depends on the error of the approximation of $Z(x)$.
The root mean squared error of this approximation can be quantified by the use of $W_N$, that is,
\begin{equation}  \label{eq: MSE}
 (\mathbb{E} \abs{W_{N}(x) - 1}^2)^{1/2} = \frac{s(x)}{\sqrt{N}}
 \qquad x\in G,\, N\in \mathbb{N},
\end{equation}
where 
\[
 s(x) := (\mathbb{E}\abs{W_{1}(x)-1}^2)^{1/2}
\]
is determined by the second moment of $W_1(x)$.
In addition, due to the appearance of the estimator $W_{N}(z)$ in the denominator of {$\widetilde r$}, 
we need some control of its distribution near zero. 
To this end, we define, for $z\in G$ and $p>0$, the inverse moment function
\[
i_{p,N}(z) :=  \left(\mathbb{E} W_{N}(z)^{-p}\right)^{\frac1p}.
\] 
With this notation we obtain the following estimate, which is proved in \ref{sec: proof_Sec4}.
\begin{lem} \label{lem: doubly_approx_err}
Assume that there exists $k\in \mathbb{N}$ such that $i_{2,k}(x)$ and $s(x)$ are finite for all $x\in G$.
Then, for all $x,z \in G$ and $N\geq k$ we have
\[
 \abs{a(x,z)-a_N(x,z)} \leq a(x,z) \frac{1}{\sqrt{N}} i_{2,k}(z)(s(x)+s(z)).
\]
\end{lem}
\begin{rem}
 One can replace the boundedness of the second inverse moment $i_{2,k}(x)$ for any $x\in G$ 
 by boundedness of a lower moment $i_{p,m}(x)$ for $p\in (0,2)$ with suitably adjusted $m\in \mathbb{N}$, 
 see Lemma~\ref{lemI} in Appendix \ref{sec: proof_Sec4}.
\end{rem}

\subsubsection{Inheritance of the Lyapunov condition}

If the second and inverse second moment are uniformly bounded, $\norm{s}_{\infty}<\infty$ as well as
$\norm{i_{2,N}}_{\infty}<\infty$, one can show that 
the Lyapunov condition of the 
MH transition kernel is inherited by the MCwM algorithm. In the following corollary, we prove this inheritance and state the resulting error bound for MCwM.

\begin{cor}\label{cor: infbounds_doubly_intr}
 For a distribution $m_0$ on $G$ let $m_n:=m_0 M_a^{n}$ and $m_{n,N}:=m_0 M_{a_N}^n$ be 
 the respective distributions of the MH and MCwM algorithms after $n$ steps.
 Let Assumption~\ref{ass: metro_V_unif} be satisfied and 
 for some $k\in \mathbb{N}$ let
 \[
  D:= 8L \norm{i_{2,k}}_\infty \norm{s}_{\infty} <\infty.
 \]
Further, define
$
 \delta_N:= \delta+D/\sqrt{N}$ and 
$
\beta_n:= n\max\{\delta_N,\alpha\}^{n-1}.
$
Then, for any 
\[
 N> \max\left\{k, \frac{D^2}{(1-\delta)^2}\right\}
\]
we have $\delta_N\in[0,1)$ and
\[
 \norm{m_n-m_{n,N}}_{\rm tv} \leq 
 \frac{DC}{\sqrt{N}} \left[m_0(V)\beta_n + \frac{L}{(1-\delta_N)(1-\alpha)} \right].
\]
\end{cor}
\begin{proof}
 Assumption~\ref{ass: metro_V_unif} 
 implies $\sup_{x\in G}\frac{M_a V(x)}{V(x)} \leq 2L$.
 By Lemma~\ref{lem: doubly_approx_err}  
 and Proposition~\ref{prop: accept_err}, with $B=G$, we obtain
 \begin{equation*} 
    \varepsilon_{V,V} = \sup_{x\in G} \frac{\norm{M_a(x,\cdot)-M_{a_N}(x,\cdot)}_{V}}{V(x)}
    \leq \frac{D}{\sqrt{N}}.
 \end{equation*}
 Further, note that
 \[
  M_{a_N} V(x)-M_aV(x) \leq \norm{M_a(x,\cdot)-M_{a_N}(x,\cdot)}_V \leq \frac{D}{\sqrt{N}} V(x),
 \]
 which implies, by Assumption~\ref{ass: metro_V_unif}, that
 for $N>D^2/(1-\delta)^2$ we have $\delta_N \in [0,1)$ 
 and $M_{a_N}V(x) \leq \delta_N V(x)+L$. 
 By Theorem~\ref{thm: est2} and Remark~\ref{rem: beta} we obtain for $r=1$ the assertion. 
\end{proof}
Observe that the estimate is bounded in $n\in \mathbb{N}$ so that the 
difference of the distributions converges uniformly in $n$ to zero for $N\to\infty$. 
The constant $\delta_N$ decreases for increasing $N$, so that larger values of $N$ improve the bound.\\

\noindent
{\bf Log-normal example I.}
Let $G=\mathbb{R}$ and the target measure $\pi$ be the standard normal distribution.
We choose a Gaussian proposal kernel $Q(x,\cdot) = \mathcal{N}(x,\gamma^2)$ for some $\gamma^2>0$, where 
$\mathcal{N}(\mu,\sigma^2)$ denotes the normal distribution with mean $\mu$ and variance $\sigma^2$.
It is well known, see \cite[Theorem~4.1, Theorem~4.3 and Theorem~4.6]{JaHa00},  
that the MH transition kernel satisfies Assumption~\ref{ass: metro_V_unif}
for some numbers $\alpha$, $C$, $\delta$ and $L$
with $V(x)=\exp( x^2/4)$.

Let $g(y;\mu,\sigma^2)$ be the density of the log-normal distribution with parameters $\mu$
and  $\sigma$, i.e., $g$ is the density of $\exp(\mu+\sigma S)$ for a random variable $S \sim \mathcal{N}(0,1)$. 
Then, by the fact that $\int_0^\infty y\, g(y; -\sigma(x)^2/2,\sigma(x)^2)\dint y=1$
for all functions 
$\sigma\colon G\to (0,\infty)$,
we can write the (unnormalized) standard normal density as
\[
 \pi_u(x) = \exp(-x^2/2) = \frac{\exp(-x^2/2)}{\int_0^\infty y\, g(y; -\sigma(x)^2/2,\sigma(x)^2)\dint y}.
\]
Hence $\pi_u$ takes the form \eqref{eq: doubly_intractable} with
$\mathcal{Y}=[0,\infty)$, $\rho(x)=\exp(-x^2/2)$, $\widebar{\rho}(x,y)=y$ and $r_x$ being a log-normal distribution
with parameters $-\sigma(x)^2/2$ and $\sigma(x)^2$. Independent draws from this log-normal distribution are used in the MCwM algorithm to 
approximate the integral. 
We have
$\mathbb{E}[W_1(x)^p]=\exp(p(p-1)\sigma(x)^2/2)$ for all $x,p\in \mathbb{R}$ and, accordingly,
\begin{align*}
 s(x) & =(\exp(\sigma(x)^2)-1)^{1/2}\leq \exp(\sigma(x)^2/2)\\
 i_{p,1}(x)& =\exp((p+1)\sigma(x)^2/2).
\end{align*}
By Lemma~\ref{lemI} we conclude that
\[
 i_{2,k}(x)\leq i_{2/k,1}(x)=\exp\left(\left( \frac{1}{2}+\frac{1}{k} \right)\sigma(x)^2\right).
\]
Hence, $\norm{s}_\infty$ as well as $\norm{i_{2,k}}_\infty$ are bounded if for some constant $c>0$ we
have $\sigma(x)^2\leq c$ for all $x\in G$.
In that case Corollary~\ref{cor: infbounds_doubly_intr} is applicable
and provides estimates for the difference between the distributions of the MH and MCwM algorithms after 
$n$-steps. However, one might ask what happens if the function $\sigma(x)^2$ is not uniformly bounded, taking, for example, 
the form $\sigma(x)^2 = \abs{x}^q$ for some $q>0$.
In Figure~\ref{fig: no_trunc_s18} we illustrate the difference of 
the distribution of the target measure to a kernel density estimator based 
on a MCwM algorithm sample for $\sigma(x)^2=\abs{x}^{1.8}$. Even though 
$s(x)$ and $i_{p,1}(x)$
grows super-exponentially in $|x|$,  
the MCwM still works reasonably well in this case.
\begin{center}
\begin{figure}[htb] 
\includegraphics[width=4.3cm]{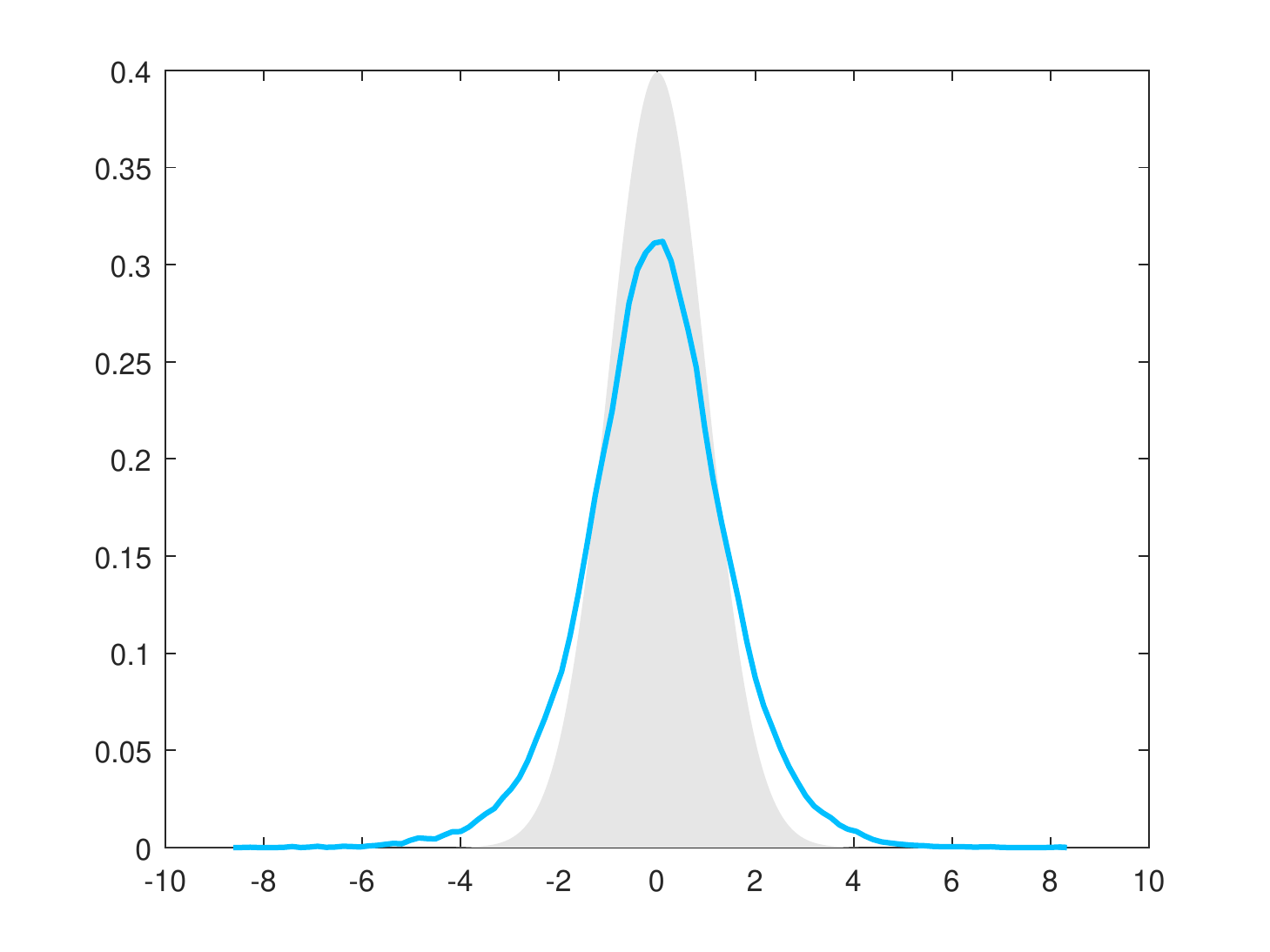}\hspace{-0.5cm}
\includegraphics[width=4.3cm]{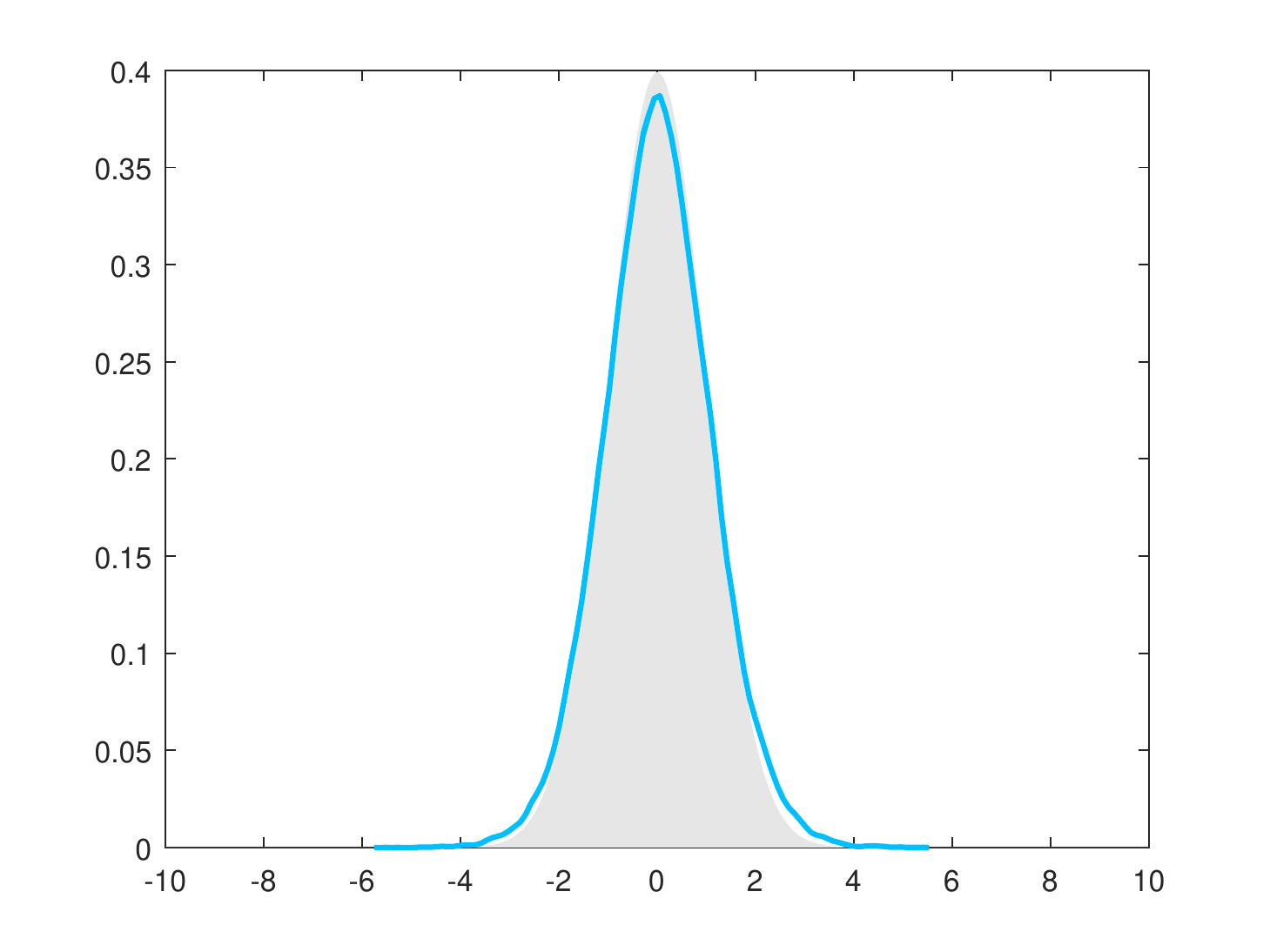}\hspace{-0.5cm}
\includegraphics[width=4.3cm]{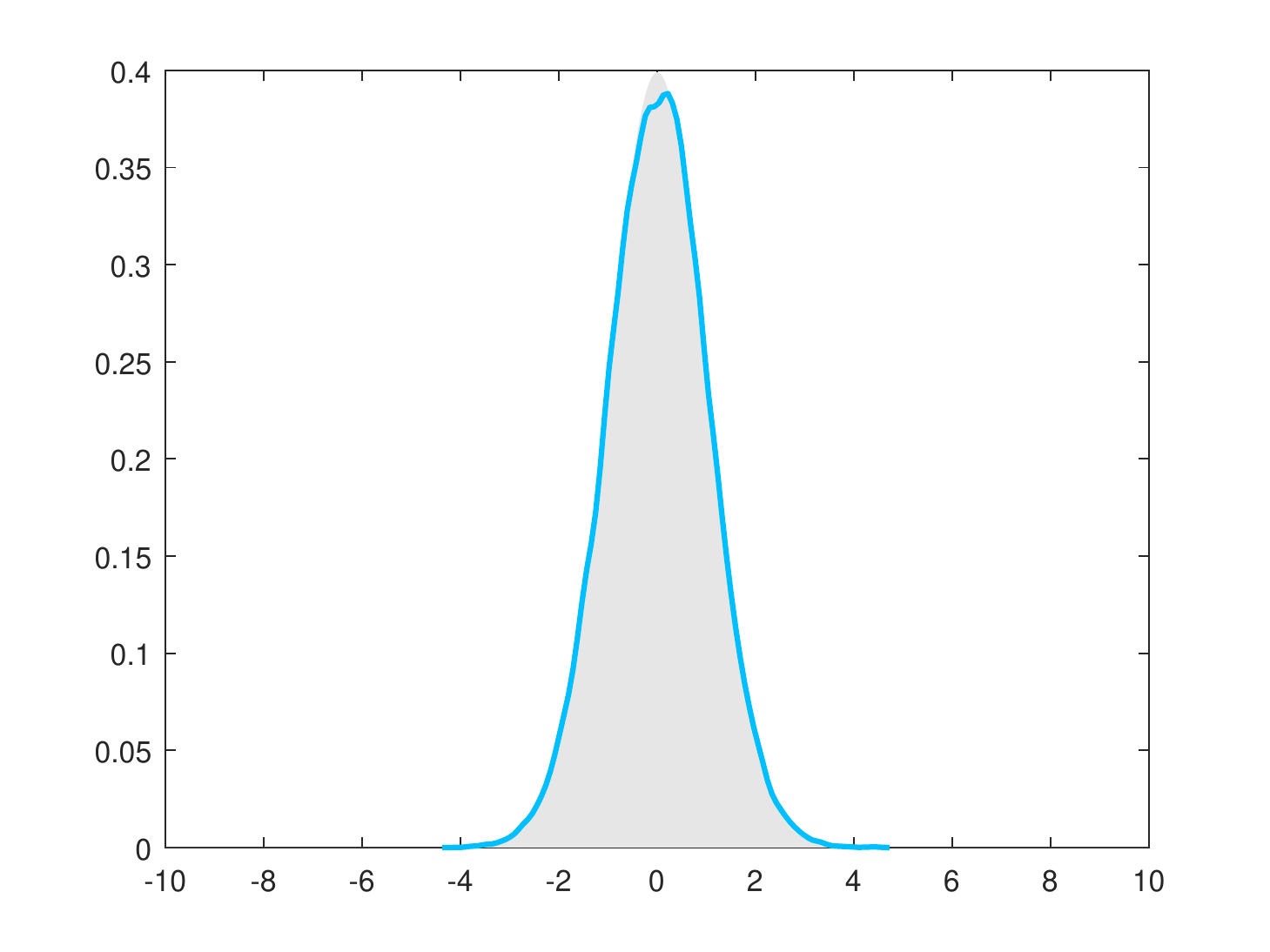}
\caption{
Here $\sigma(x)^2:=\abs{x}^{1.8}$ for $x\in \mathbb{R}$.
The target density (standard normal) is plotted in grey, a kernel density estimator based on $10^5$
steps of the MCwM algorithm with $N=10$ (left), $N=10^2$ (middle)  and 
$N=10^3$ (right) is plotted in blue.
} 
\label{fig: no_trunc_s18}
\end{figure}
\end{center}
However, in Figure~\ref{fig: no_trunc_s22} we consider 
the case where $\sigma(x)^2=\abs{x}^{2.2}$ and the behavior changes dramatically. Here the MCwM algorithm
does not seem to work at all. This motivates a modification
of the MCwM algorithm in terms of restricting the state space to 
the ``essential part'' determined by the Lyapunov condition.
\begin{center}
\begin{figure}[htb]  
\includegraphics[width=4.3cm]{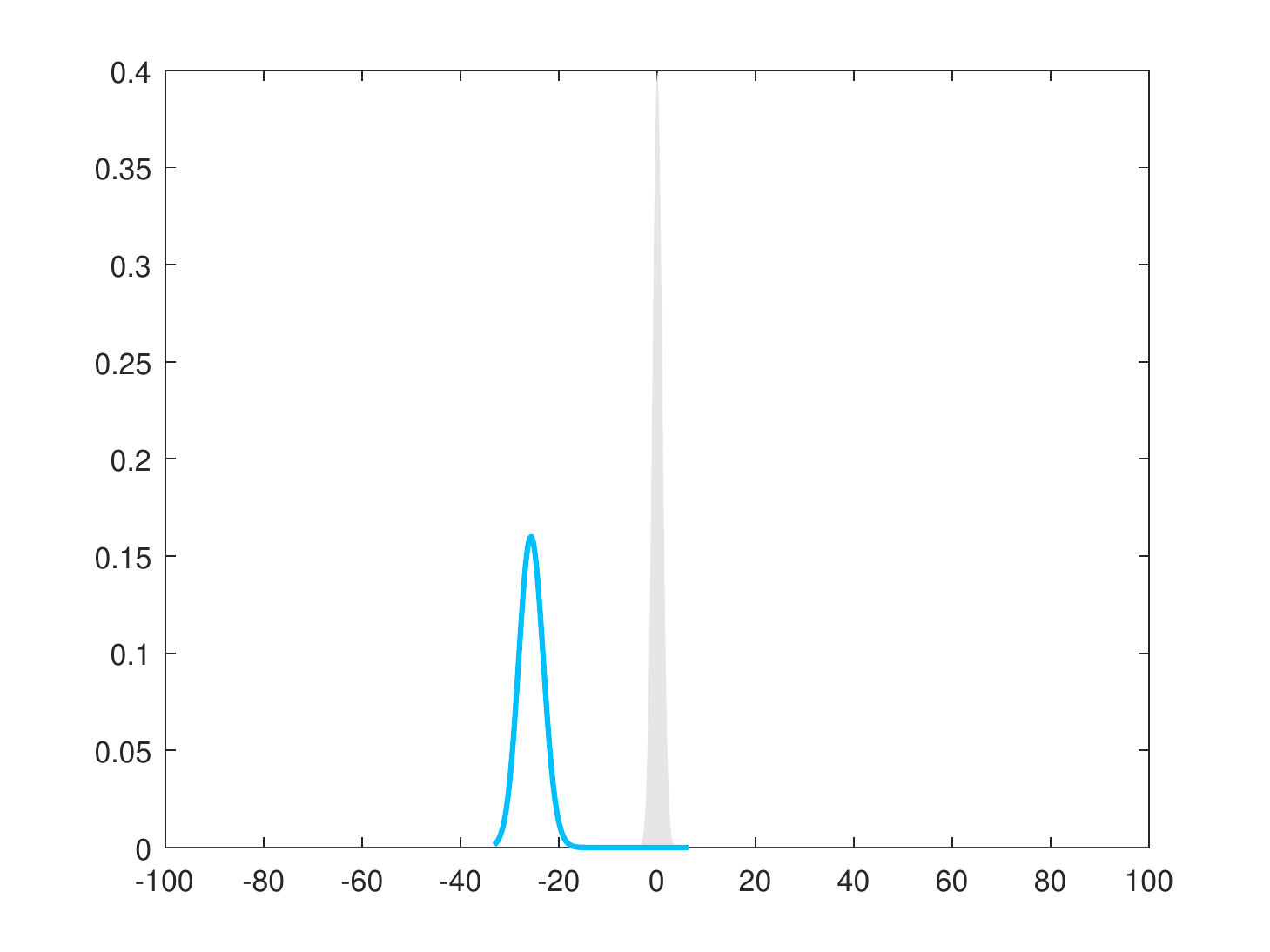}\hspace{-0.5cm}
\includegraphics[width=4.3cm]{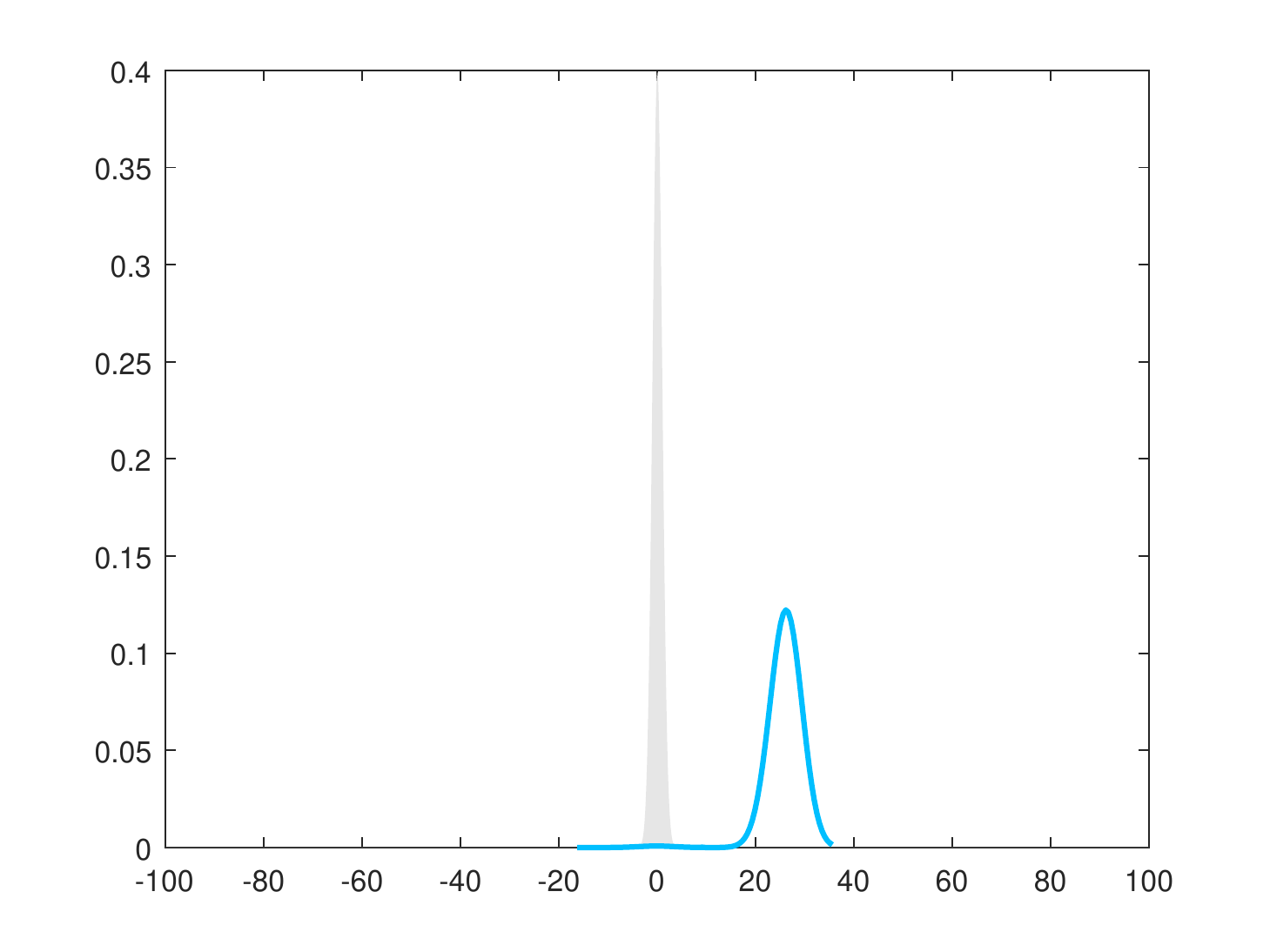}\hspace{-0.5cm}
\includegraphics[width=4.3cm]{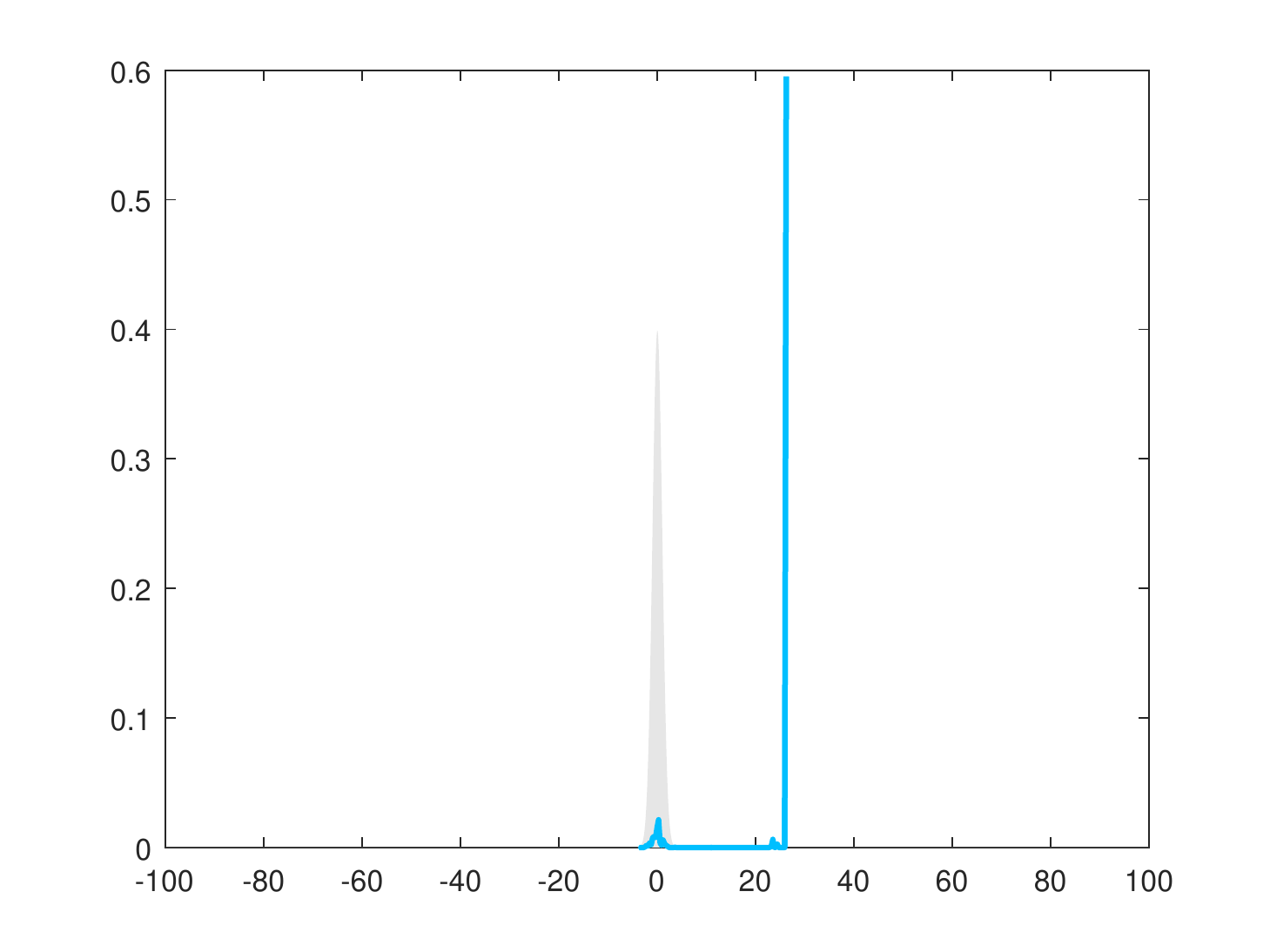}
\caption{
Here $\sigma(x)^2:=\abs{x}^{2.2}$ for $x\in \mathbb{R}$.
The target density (standard normal) is plotted in grey, a kernel density estimator based on $10^5$
steps of the MCwM algorithm with $N=10$ (left), $N=10^2$ (middle)  and 
$N=10^3$ (right) is plotted in blue.
} 
 \label{fig: no_trunc_s22}
\end{figure}
\end{center}

\subsubsection{Restricted MCwM approximation}
\label{sec: rest_mcwm_approx}

With the notation and definition from the previous section we consider
the case where the functions 
$i_{2,k}(x)$ and $s(x)$ are not uniformly bounded. 
Under Assumption~\ref{ass: metro_V_unif} there are two simultaneously used tools 
which help to control the difference of a 
transition of MH and MCwM:
\begin{enumerate}
 \item The Lyapunov condition leads
 to a weight function and eventually to a weighted norm, see Proposition~\ref{prop: accept_err}.
 \item By restricting the MCwM to the ``essential part'' of the state space we
 prevent that the approximating Markov chain deteriorates. Namely, for some $R\geq 1$ we restrict the
 MCwM to $B_R$, see Section~\ref{sec: restr_approx}.
\end{enumerate}
For $x,z\in G$
the {acceptance ratio $\widetilde r$ 
used in Algorithm~\ref{alg: MCwM_doubly_intract}
} is now modified to{
\[
 \mathbf{1}_{B_R}(z) \cdot \widetilde r(x,z,W_N(x),W_N(z))
\]}
which leads to the \emph{restricted MCwM algorithm}:
\begin{alg} \label{alg: MCwM_doubly_intract_restr}
For given $R\geq 1$ and a proposal transition kernel $Q$ a transition from $x$
to $y$ of the restricted MCwM algorithm works as follows.
\begin{enumerate}
 \item Draw $U\sim \text{Unif}[0,1]$ and a proposal $Z\sim Q(x,\cdot)$ independently, 
 call the result $u$ and $z$, respectively.
\item {Calculate $\widetilde r(x,z,W_N(x),W_N(z))$ based on independent samples for $W_N(x)$, $W_N(z)$, which are also independent from previous iterations.}
 \item If {$u< \mathbf{1}_{B_R}(z)\cdot \widetilde r(x,z,W_N(x),W_N(z))$}, then accept the proposal, and return $y:=z$, 
 otherwise reject the proposal and return $y:=x$. 
\end{enumerate}
\end{alg}
Given the current state $x\in G$ and a proposed state $z\in G$ 
the overall acceptance probability is{
\begin{align*}
 a^{(R)}_{N}(x,z)&:=\mathbb{E}\left[ \min\left\{1, \mathbf{1}_{B_R}(z) \cdot \widetilde r(x,z,W_N(x),W_N(z)) \right\} \right] = \mathbf{1}_{B_R}(z)\cdot a_N(x,z),
\end{align*}}
which leads to
the corresponding transition kernel of the form $M_{a^{(R)}_{N}}$, see \eqref{eq: MH_type}.
By using Theorem~\ref{thm: restricted_approx} and Proposition~\ref{prop: accept_err}
we obtain the following estimate. 
\begin{cor}  \label{cor: doub_restr_MCwM}
 Let Assumption~\ref{ass: metro_V_unif} be satisfied, i.e., $M_a$ is $V$-uniformly ergodic and
 the function $V$ as well as the constants $\alpha, C, \delta$ and $L$ are determined.
 For $\beta\in (0,1)$ and $R\geq 1$ let
 \begin{align*}
  B_R & := \left\{ x\in G\mid V(x)\leq R \right\},\\
  D_R & := 12\cdot L \norm{i_{2,k}\cdot\mathbf{1}_{B_R}}_{\infty,V^{1-\beta}} \norm{s\cdot\mathbf{1}_{B_R}}_{\infty,V^{\beta}}
	  <\infty.
 \end{align*}
 Let $m_0$ be a distribution on $B_R$ and $\kappa := \max\{ m_0(V),L/(1-\delta)\}$.
 Then, for
 \begin{equation} \label{eq: N_restr_doubly_intr}
  N\geq \max\left\{k,4 \left(\frac{R\cdot D_R}{1-\delta}\right)^2 \right\}
 \end{equation}
and $R\geq \exp(1)$ we have
 \[
  \norm{m_n-m^{(R)}_{n,N}}_{\rm tv} \leq \frac{33C(L+1) \kappa}{1-\alpha} \cdot
  \frac{\log R}{R},
 \]
 where $m^{(R)}_{n,N}:= m_0 M^n_{a_N^{(R)}}$ and $m_n:= m_0 M^n_{a}$ are the 
 distributions of the MH and restricted MCwM algorithm after $n$-steps.
\end{cor}
\begin{proof}
 We apply Theorem~\ref{thm: restricted_approx} with $P(x,\cdot)=M_a(x,\cdot)$ and
 \[
  \widetilde P(x,\cdot)
  = \mathbf{1}_{B_{R}}(x)\, M_{a_N^{(R)}}(x,\cdot) + \mathbf{1}_{B_{R}^c}(x) \delta_{x_0}(\cdot), 
  \quad x\in G,
 \]
 for some $x_0\in B_{R}$. Note that $\wt P(x,B_{R})=1$ for any $x\in G$. Further $\wt P$
 and $M_{a_N^{(R)}}$ coincide on $B_{R}$,
thus we also have 
 $\wt P^n=M^n_{a_N^{(R)}}$ on $B_{R}$ for $n\in \mathbb{N}$.
 Observe also that the restriction of $P$ to $B_{R}$, denoted by $P_{R}$, satisfies $P_{R} = M_{a^{({R})}}$ 
 with $a^{(R)}(x,z) := \mathbf{1}_{B_R}(z)\, a(x,z)$.
 Hence
 \begin{align*}
  \Delta({R}) 
 & = \sup_{x\in B_{{R}}} \frac{\norm{M_{a^{(R)}}(x,\cdot)-M_{a_N^{(R)}}(x,\cdot)}_{\rm tv}}{V(x)}.
 \end{align*}
Moreover, we have by Lemma~\ref{lem: doubly_approx_err} that
 \begin{align*}
    \abs{a^{(R)}(x,z)-a_N^{(R)}(x,z)}
  = & \mathbf{1}_{B_R}(z) \abs{a(x,z)-a_N(x,z)}\\ 
  \leq & \mathbf{1}_{B_R}(z) \cdot a(x,z) \frac{1}{\sqrt{N}}\, i_{2,k}(z)(s(x)+s(z))\\
  = & a^{(R)}(x,z) \frac{1}{\sqrt{N}}\, i_{2,k}(z)(s(x)+s(z)).
 \end{align*}
 With Proposition~\ref{prop: accept_err} and
 \[
 \sup_{x\in G}\frac{M_{a^{(R)}}V(x)}{V(x)} \leq \sup_{x\in G}\frac{M_{a}V(x)}{V(x)}+1
 \underset{\rm Ass.~\ref{ass: metro_V_unif}}{\leq} 3L,
 \]
we have that
 $\Delta(R)\leq D_R/\sqrt{N}$. 
Then, by $N\geq 4 (RD_R/(1-\delta))^2$ we obtain
 \[
  R\cdot\Delta(R) \leq \frac{1-\delta}{2}
 \]
 such that all conditions of Theorem~\ref{thm: restricted_approx} are verified and the stated estimate follows.
\end{proof}
\begin{rem}
  The estimate depends crucially on the sample size $N$ as well as on the parameter $R$. 
  If the influence of $R$ in $D_R$ is explicitly known, then one can choose $R$ depending on $N$ in such away
  that the conditions of the corollary are satisfied and one eventually obtains an upper bound on the 
  total variation distance of the difference between the distributions depending only on $N$ and not on $R$ anymore. {For example, if we additionally assume that the function $g\colon(0,\infty)\to (0,\infty)$ given by $g(R) = R\cdot D_R$ is invertible, then for $N\geq k$ and the choice $R:= g^{-1}\left((1-\delta)\sqrt{N}/2\right)$ we have 
\[
\norm{m_n-m^{(R)}_{n,N}}_{\rm tv} \leq \frac{33C(L+1) \kappa}{1-\alpha} \cdot
\frac{\log\left(g^{-1}\left((1-\delta)\sqrt{N}/2\right)\right)}{g^{-1}\left((1-\delta)\sqrt{N}/2\right)}.
\]
 Thus, depending on whether and how fast $g^{-1}\left((1-\delta)\sqrt{N}/2\right) \to \infty$ for $N\to \infty$ determines the convergence of the upper bound of $\norm{m_n-m^{(R)}_{n,N}}_{\rm tv}$ to zero.} 
\end{rem}
\noindent
{\bf Log-normal example II.}
We continue with the log-normal example. In this setting we have
\begin{align*}
 B_R & = \{ x\in \mathbb{R} \mid  \abs{x}\leq2 \sqrt{ \log R}\},\\
 \norm{i_{2,k}\cdot \mathbf{1}_{B_R}}_{\infty, V^{1-\beta}} 
 & \leq \sup_{\abs{x}\leq 2 \sqrt{\log R}} \exp\left( \left(\frac{1}{2}+\frac{1}{k}\right)\sigma(x)^2 - 
	\frac{1-\beta}{4} x^2 \right),\\
 \norm{s\cdot \mathbf{1}_{B_R}}_{\infty, V^{\beta}} 
 & \leq \sup_{\abs{x}\leq 2 \sqrt{\log R}} \exp\left(\sigma(x)^2/2 - \beta x^2/4 \right).
\end{align*}
{Thus, $D_R$ is uniformly bounded in $R$ for $\sigma(x)^{2}\propto|x|^{q}$ with $q<2$ and not uniformly bounded for $q>2$. }
As in the numerical experiments in Figure~\ref{fig: no_trunc_s18} and Figure~\ref{fig: no_trunc_s22} 
let us consider the cases $\sigma(x)^2 = \abs{x}^{1.8}$ and $\sigma(x)^2 = \abs{x}^{2.2}$.
In Figure~\ref{fig: trunc_s18} we compare the normal target density
with a kernel density estimator based on the restricted MCwM on $B_R=[-10,10]$
and observe essentially the same reasonable behavior as in Figure~\ref{fig: no_trunc_s18}. 
In Figure~\ref{fig: trunc_s22} we consider the same scenario and observe that
the restriction indeed stabilizes. In contrast to Figure~\ref{fig: no_trunc_s22}, 
convergence to the true target distribution is visible but, in line with the theory, slower than for $\sigma(x)^2 = \abs{x}^{1.8}$. 
\begin{center}
\begin{figure}[htb] 
\includegraphics[width=4.3cm]{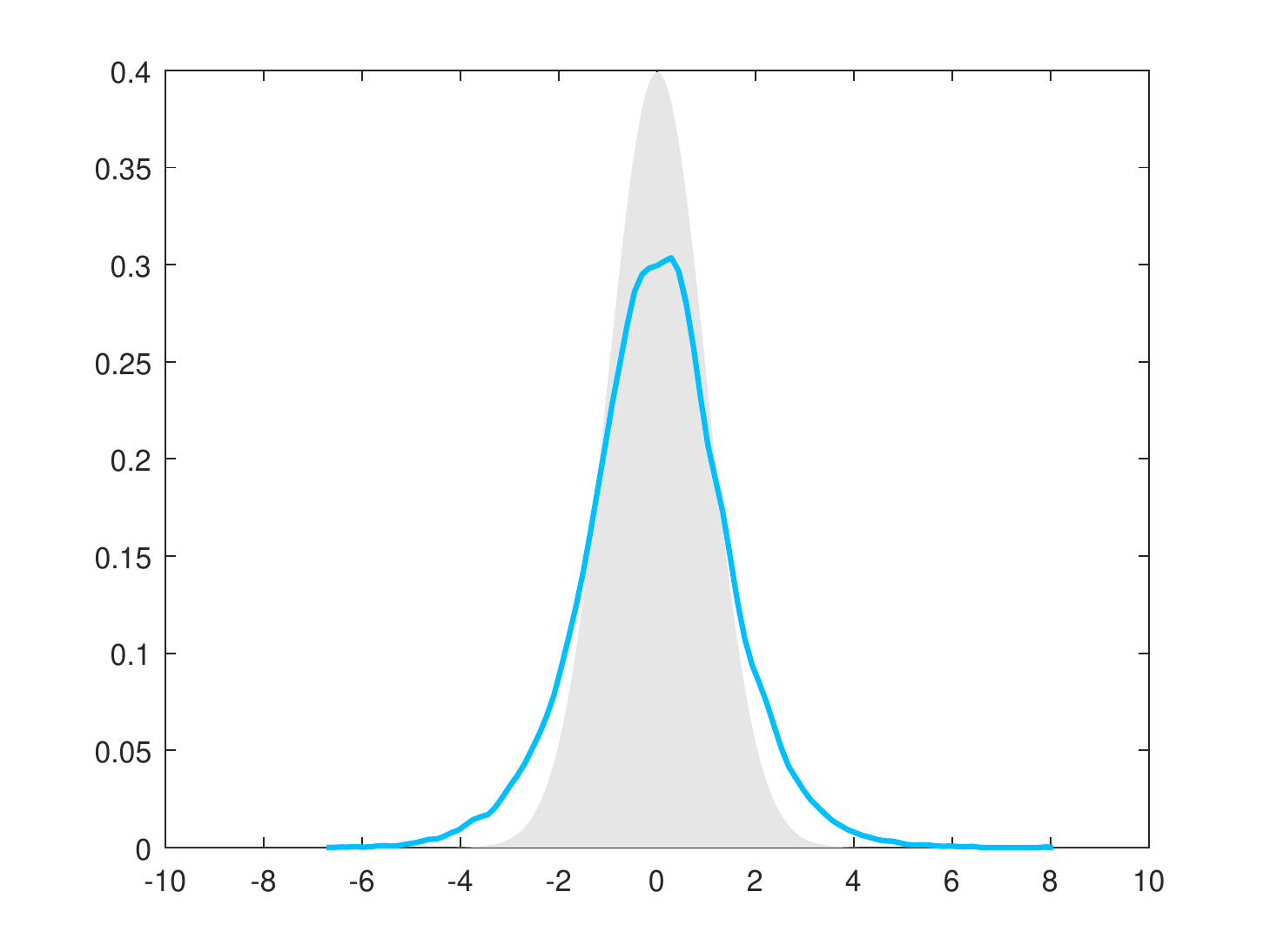}\hspace{-0.5cm}
\includegraphics[width=4.3cm]{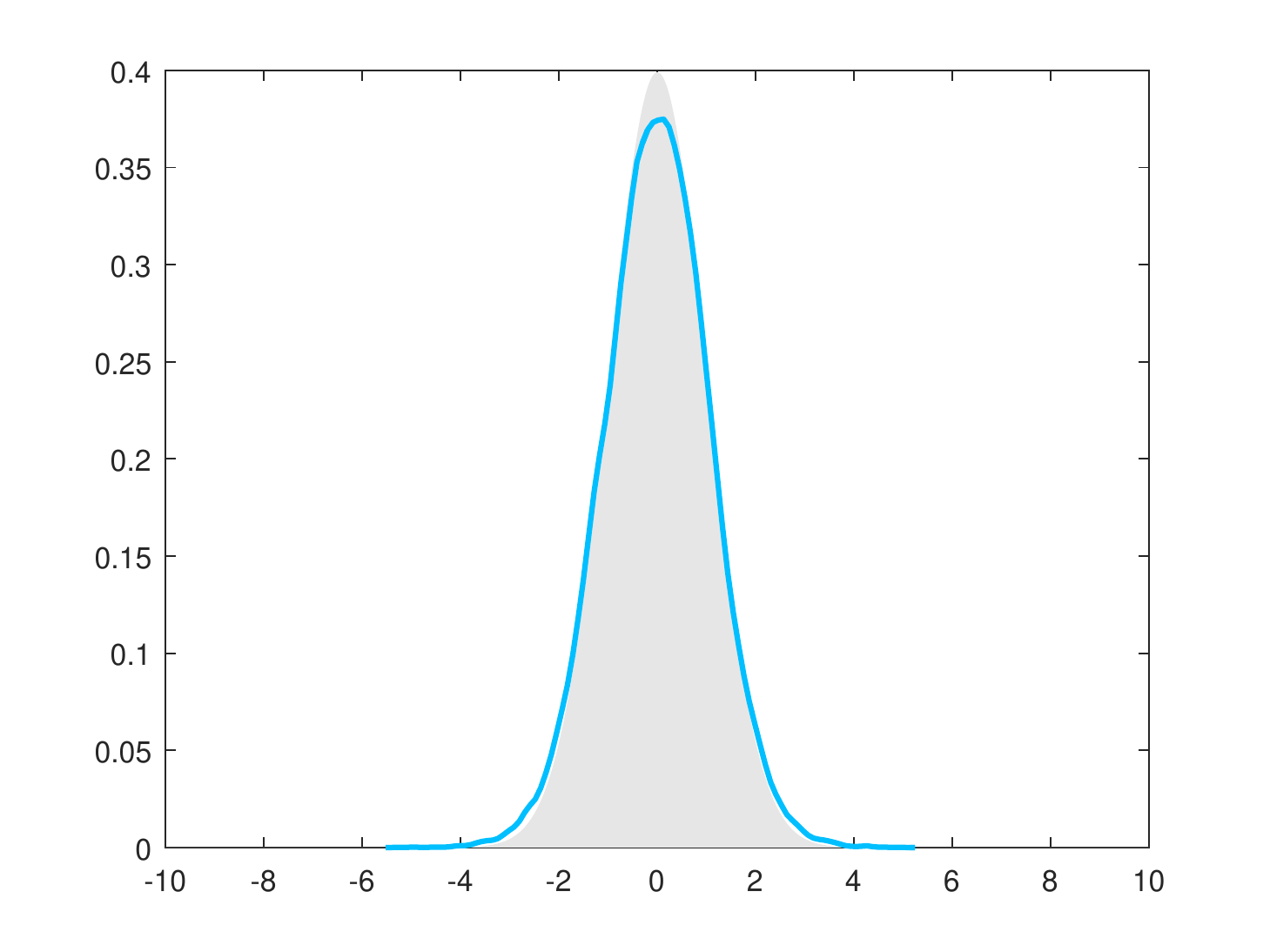}\hspace{-0.5cm}
\includegraphics[width=4.3cm]{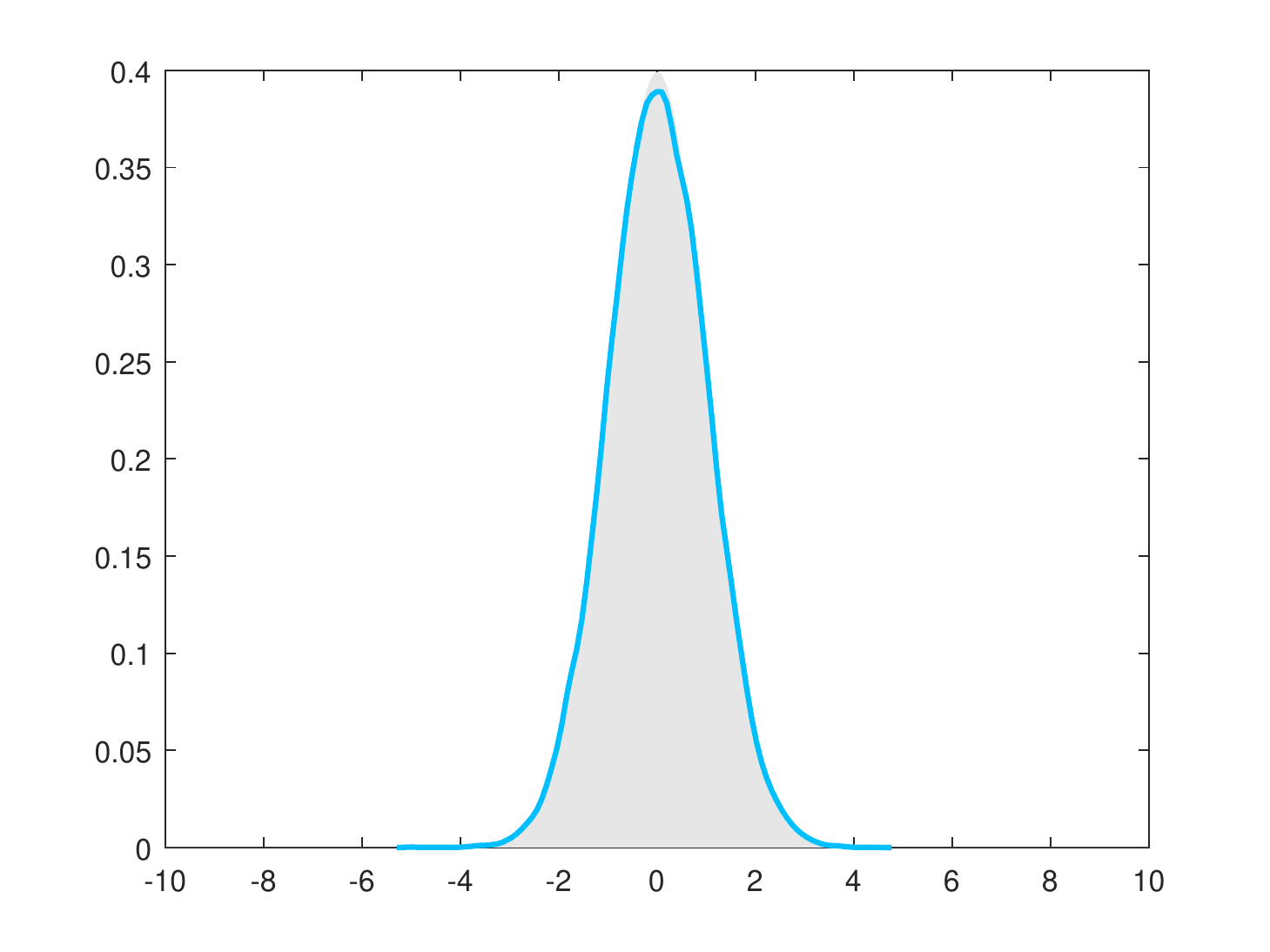}
\caption{
Here $\sigma(x)^2:=\abs{x}^{1.8}$ for $x\in \mathbb{R}$ and $B_R=[-10,10]$.
The target density (standard normal) is plotted in grey, a kernel density estimator based on $10^5$
steps of the MCwM algorithm with $N=10$ (left), $N=10^2$ (middle)  and 
$N=10^3$ (right) is plotted in blue.
} 
\label{fig: trunc_s18}
\end{figure}
\end{center}
\begin{center}
\begin{figure}[htb] 
\includegraphics[width=4.3cm]{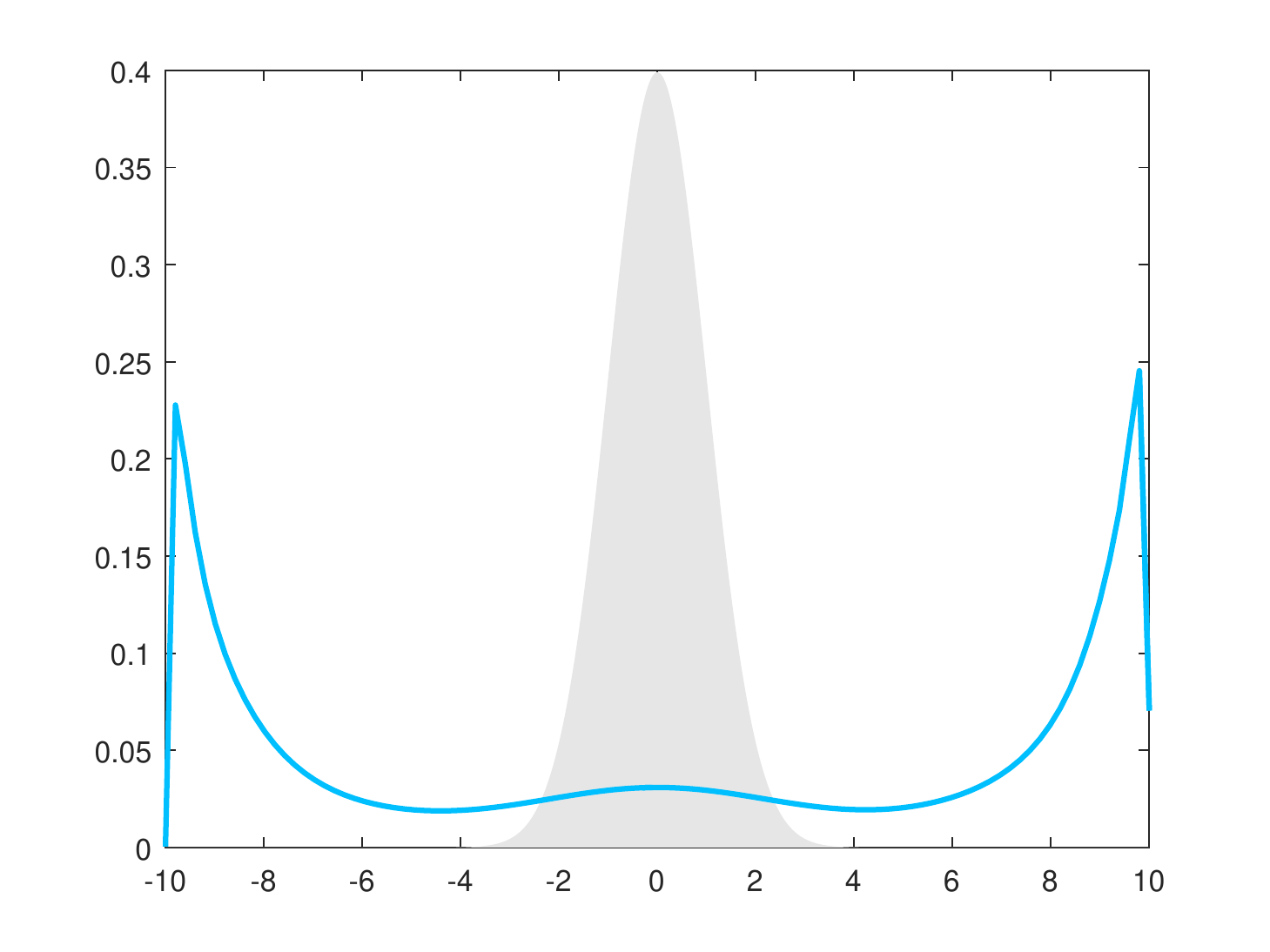}\hspace{-0.5cm}
\includegraphics[width=4.3cm]{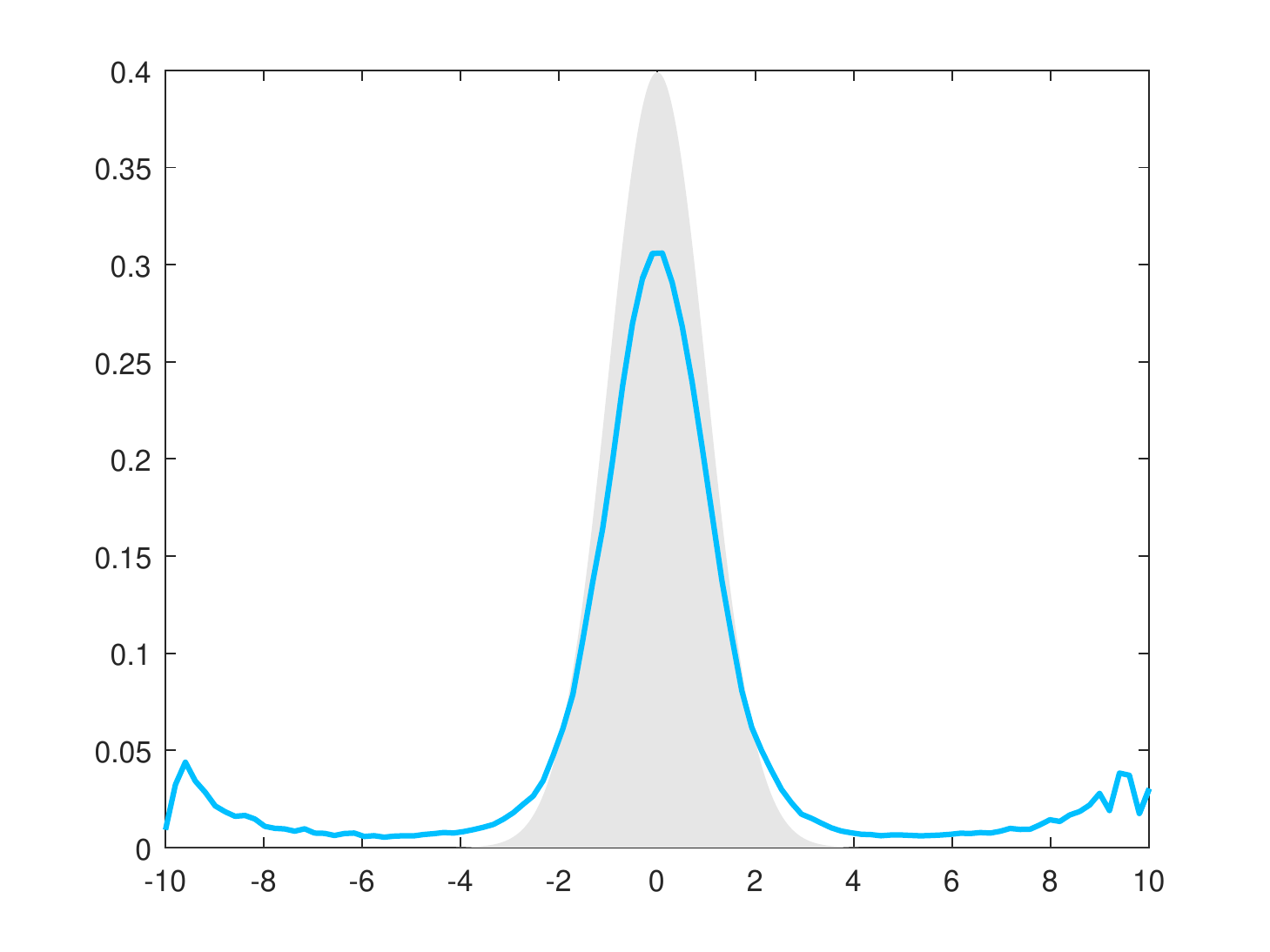}\hspace{-0.5cm}
\includegraphics[width=4.3cm]{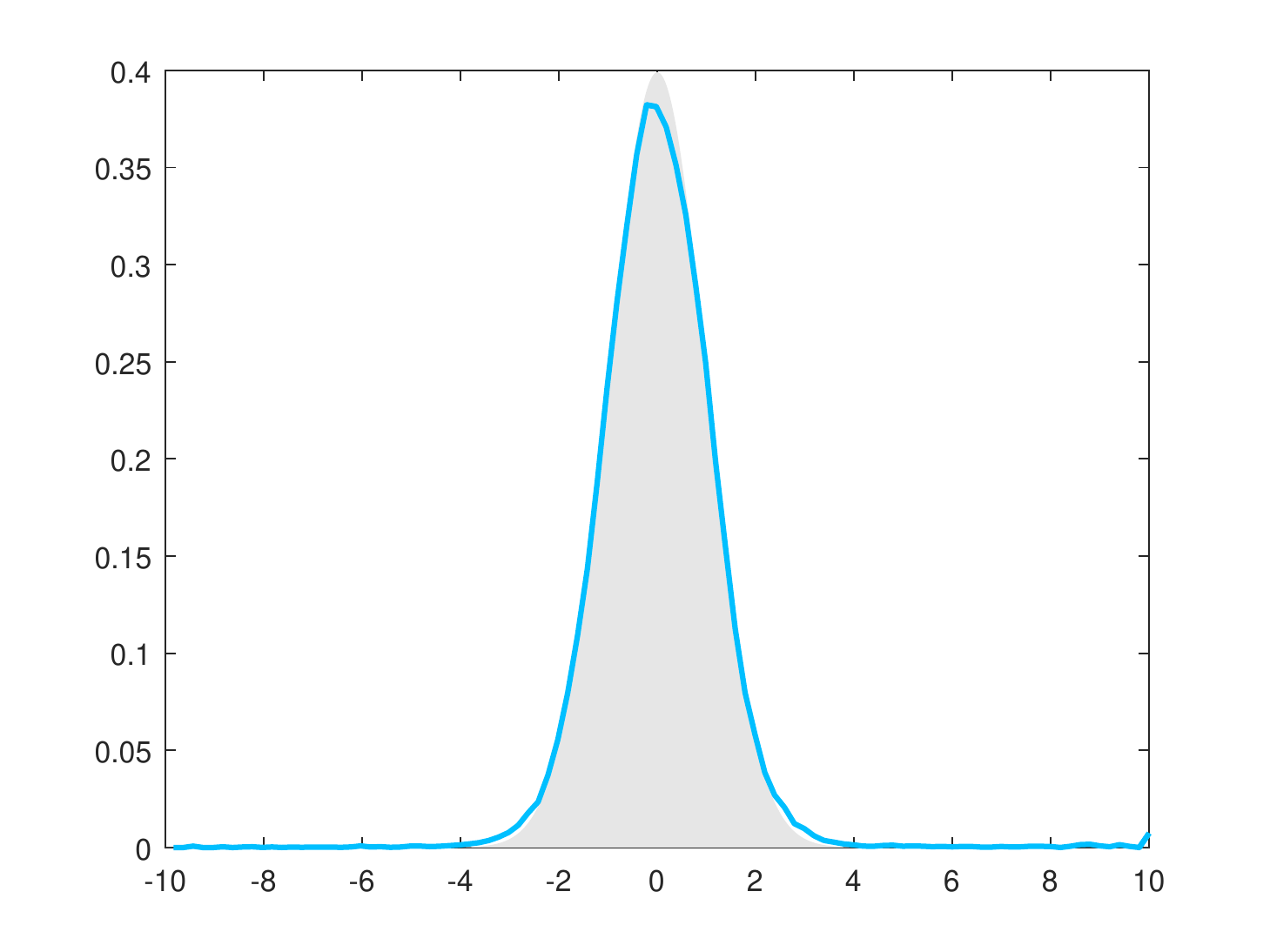}
\caption{
Here $\sigma(x)^2:=\abs{x}^{2.2}$ for $x\in \mathbb{R}$ and $B_R=[-10,10]$.
The target density (standard normal) is plotted in grey, a kernel density estimator based on $10^5$
steps of the MCwM algorithm with $N=10$ (left), $N=10^2$ (middle)  and 
$N=10^3$ (right) is plotted in blue.
} 
\label{fig: trunc_s22}
\end{figure}
\end{center}
Now we apply Corollary~\ref{cor: doub_restr_MCwM} in both cases {and note that by similar arguments as below one can also treat $\sigma(x)^2\propto\abs{x}^{q}$ with, respectively, $q<2$ or $q>2$.}

{\bf 1. Case $\sigma(x)^2=\abs{x}^{1.8}$.} For $k=100$ and $\beta=1/2$ one can easily see 
that $\norm{i_{2,100}\cdot \mathbf{1}_{B_R}}_{\infty,V^{1/2}}$ and $\norm{s\cdot \mathbf{1}_{B_R}}_{\infty,V^{1/2}}$
is bounded by $6000$, independent of $R$. Hence there is a constant $D\geq 1$ so that $D_R \leq D$. 
With this knowledge we choose
$R=\frac{(1-\delta)}{\sqrt{2}D} \sqrt{N}$ such that for 
$N\geq\max\left\{100,\frac{2 \exp(2) D^2}{(1-\delta)^2}\right\}$ condition 
\eqref{eq: N_restr_doubly_intr} and $R\geq\exp(1)$ is satisfied.
Then, Corollary~\ref{cor: doub_restr_MCwM} gives the existence of a constant $\widetilde C>0$, so that
\[
  \norm{m_n-m^{(R)}_{n,N}}_{\rm tv} \leq \widetilde C\;
  \frac{\log N}{\sqrt{N}}
\]
for any initial distribution $m_0$ on $B_R$. 

{\bf 2. Case $\sigma(x)^2=\abs{x}^{2.2}$.} 
For $k=100$ and $\beta=1/2$ we obtain
\begin{align*}
 \norm{i_{2,100}\cdot \mathbf{1}_{B_R}}_{\infty,V^{1/2}} 
 & \leq \exp\left( 2.5 \left( \log R\right)^{11/10}  \right),\\
 \norm{s\cdot \mathbf{1}_{B_R}}_{\infty,V^{1/2}} 
 & \leq \exp\left( 2.5 \left( \log R\right)^{11/10}  \right).
\end{align*}
Hence 
$
 D_R \leq 12 L \exp\left( 5 \left( \log R\right)^{11/10}  \right).
$
Eventually, for 
\[
N\geq \max\left\{ 100, \frac{24^2 \exp(2\cdot 6^{11/10}) L^2}{(1-\delta)^2} \right\}
\]
we have with
$R= \exp\left(  \frac{1}{6} \left[ \log\left( \frac{\sqrt{N}(1-\delta)}{24L} \right) \right]^{10/11} \right)$
that $R\geq \exp(1)$ and \eqref{eq: N_restr_doubly_intr} is satisfied.
Then, with
$ \widetilde C_1 := \frac{33C(L+1) \kappa}{1-\alpha}$,
$  \widetilde C_2 := \sqrt{\frac{1-\delta}{24L}}$
and Corollary~\ref{cor: doub_restr_MCwM} we have
\[
 \norm{m_n-m^{(R)}_{n,N}}_{\rm tv} 
 \leq \frac{\widetilde C_1  \cdot \frac{1}{6\cdot 2^{10/11}}\left[ \log\left( \widetilde C_2 N\right) \right]^{10/11}}
 {\exp\left(\frac{1}{6\cdot 2^{10/11}}\left[ \log\left( \widetilde C_2 N\right) \right]^{10/11} \right)}
 \leq \frac{\widetilde{C}_1 (k+1)!}{\left[ \log\left( \widetilde C_2 N \right)\right]^{10k/11}},\quad   
\]
for any initial distribution $m_0$ on $B_R$ and all $k\in\mathbb{N}$. Here
the last inequality follows by the fact that
$\exp(x)\geq \frac{x^{k+1}}{(k+1)!}$ for any $x\geq 0$ and $k\in\mathbb{N}$. 

To summarize, by suitably choosing $N$ and $R$ (possibly depending on $N$) sufficiently large 
the difference between the distributions of the restricted MCwM and the MH algorithms after $n$-steps can be made arbitrarily small.

\subsection{Latent variables}\label{sec: latvar}

In this section we consider $\pi_u$ of the form \eqref{eq: latent}. 
Here, as for doubly intractable distributions, 
the idea is to substitute $\pi_u(x)$ in the acceptance probability
of the MH algorithm by a Monte Carlo estimate
\[
 \widehat \rho_N(x) = \frac{1}{N} \sum_{i=1}^N 
  \widebar{\rho}(x,Y^{(x)}_i)
\]
where we assume that we have access to an iid 
sequence of random variables $(Y^{(x)}_i)_{1\leq i\leq N}$ 
where each $Y^{(x)}_i$
has distribution $r_x$. Define
a function $W_N\colon G\to \mathbb{R}$
by $W_{N}(x) := \widehat{\rho}_N(x)/\pi_u(x)$ 
and note that $\mathbb{E}[W_{N}(x)]=1$. 
Then, the acceptance probability given $W_N(x)$, $W_N(z)$
modifies to{
\[
 a_N (x,z) :=\mathbb{E} \left[ \min \left\{1,r(x,z)\cdot\frac{W_N(z)}{W_N(x)} \right\} \right]
\]
where} $W_N(x)$, $W_N(z)$ are assumed to be independent random variables.
Note that all the objects which depend on {$a_N$, such
as $M_{a_N}, a_{N}^{(R)}, M_{a_{N}^{(R)}}$}, 
that appear in this section are defined just as in Section~\ref{sec: doubly_intractable}. The only difference
is that the order of the {variables $W_N(x)$ and $W_N(z)$ in the ratio $\widetilde r$ at  \eqref{eq: doubl_intr_a_N} has been reversed}.
Thus, this leads to a MCwM algorithm as stated in Algorithm~\ref{alg: MCwM_doubly_intract}, where
the transition kernel is given by $M_{a_N}$.

Also as in Section~\ref{sec: doubly_intractable} we define 
$s(x) := \left( \mathbb{E}\abs{W_1(x)-1}^2 \right)^{1/2}$ 
and $ i_{p,N}(x) := (\mathbb{E}W_N(x)^{-p})^{1/p}$ for all $x\in G$ and $p>0$.
With those quantities we obtain the following estimate of the difference of the acceptance
probabilities of $M_a$ and $M_{a_N}$ proved in \ref{sec: proof_Sec4}.
\begin{lem} 
\label{lem: diff_acc_prob1}
Assume that there exists $k\in \mathbb{N}$ 
such that $i_{2,k}(x)$ and $s(x)$ are finite for all $x \in G$. 
Then, for all $x,z \in G$ and $N\geq k$ we have 
\begin{align}  \label{eq: MCWM_relative}
 \abs{a(x,z)- a_N(x,z) }
& \leq  a(x,z)\,
 \frac{1}{\sqrt{N}} \, i_{2,k}(x) (s(x)+s(z)).
\end{align}
\end{lem}

If $\norm{s}_{\infty}$ and $\norm{i_{2,k}}_{\infty}$ are finite for some $k\in \mathbb{N}$, then the 
same statement as formulated in Corollary~\ref{cor: infbounds_doubly_intr} holds. 
The proof works exactly as stated there. 
Examples which satisfy this condition are for instance presented in \cite{MeLeRo17}.
However, there are cases where the functions $s$ and $i_{2,k}$ are unbounded. 
In this setting, as in Section~\ref{sec: rest_mcwm_approx}, we consider the restricted MCwM
algorithm with transition kernel $M_{a_N^{(R)}}$.
Here again the same statement and proof as formulated in Corollary~\ref{cor: doub_restr_MCwM} hold.
We next provide an application of this corollary in the latent variable setting.\\

\noindent
{\bf Normal-normal model.}
Let $G=\mathbb{R}$ and the function $\varphi_{\mu,\sigma^2}$ be the density of $\mathcal{N}(\mu,\sigma^2)$.
For some $z\in \mathbb{R}$ and (precision) parameters $\gamma_Z,\gamma_Y>0$ define
\[
  \pi_u(x) := \int_\mathbb{R} \varphi_{z,\gamma_Z^{-1}}(y)\, \varphi_{0,\gamma_Y^{-1}}(x-y) \dint y, 
\]
that is, $\mathcal{Y}=\mathbb{R}$, $\widebar \rho(x,y) = \varphi_{z,\gamma_Z^{-1}}(y)$ and $r_x = \mathcal{N}(x,\gamma_Y^{-1})$.
By the convolution of two normals the target distribution $\pi$ satisfies
\begin{equation} \label{eq: normal_normal_pi_u}
 \pi_u(x) = \varphi_{z,\gamma_{Z,Y}^{-1}}(x), 
 \quad \text{with} \quad
 \gamma_{Z,Y}^{-1} := \gamma_Z^{-1}+\gamma_Y^{-1}.
\end{equation}
Note that, for real-valued random variables $Y,Z$ the probability measure $\pi$ is the posterior 
distribution given an observation $Z=z$ within the model
\begin{align*}
Z|Y=y & \sim\mathcal{N}\left(y,\gamma_{Z}^{-1}\right),\qquad 
Y|x\sim\mathcal{N}\left(x,\gamma_{Y}^{-1}\right),
\end{align*}
with the improper Lebesgue prior imposed on $x$.

Pretending that we do not know $\pi_u(x)$ we compute 
\[
 \widehat{\rho}_N(x) = \frac{1}{N} \sum_{i=1}^N \varphi_{z,\gamma_Z^{-1}}(Y_i^{(x)}),
\]
where $(Y_i^{(x)})_{1\leq i \leq N}$ is a sequence of iid random variables with $Y_1^{(x)}\sim \mathcal{N}(x,\gamma_Y^{-1})$.
Hence
\begin{align*}
W_{N}(x) 
=\frac{1}{N}\sum_{i=1}^{N}\frac{\varphi_{z,\gamma_Z^{-1}}(Y_i^{(x)})}{\varphi_{z,\gamma_{Z,Y}^{-1}}(x)}
& =\frac{1}{N}\left(\frac{\gamma_{Z}}{\gamma_{Z,Y}}\right)^{1/2}
\sum_{i=1}^{N}
\frac{\varphi_{0,1}(\sqrt{\gamma_{Z}}(z-Y_{i}^{(x)}))}
     {\varphi_{0,1}(\sqrt{\gamma_{Z,Y}}(z-x))}.
\end{align*}
By using a random variable $\xi \sim \mathcal{N}(0,1)$ we have
for $p>-\gamma_Y/\gamma_Z$ that
\begin{align}
\mathbb{E}\left[W_{1}(x)^{p}\right] & =\left(\frac{\gamma_{Z}}{\gamma_{Z,Y}}\right)^{p/2}
\mathbb{E}\left[\exp\left(\frac{p}{2}\gamma_{Z,Y}\left(z-x\right)^{2} 
-\frac{p}{2}\,\frac{\gamma_{Z}}{\gamma_{Y}}(\gamma_{Y}^{1/2}(z-x)-\xi)^{2}\right) \right]\nonumber \\
 & \propto\exp\left(\frac{\gamma_{Z}\,\gamma_{Z,Y}\,p\left(p-1\right)}{2
 \left(\gamma_{Y}+p\gamma_{Z}\right)}\left(z-x\right)^{2}\right) .\label{eq:moment_W_example}
\end{align}
Here $\propto$ means equal up to a constant independent of $x$.
As a consequence, $\left\Vert s\right\Vert _{\infty}=\infty$ and therefore 
Corollary~\ref{cor: infbounds_doubly_intr} (which is also true in the latent variable setting) cannot be applied. 
Nevertheless, we can obtain bounds for the restricted MCwM in this example using the statement of 
Corollary~\ref{cor: doub_restr_MCwM}
by controlling $s$ and $i_{2,k}$ using a Lyapunov function $V$. 
The following result, proved in \ref{sec: proof_Sec4}, verifies the necessary moment conditions under some additional restrictions on the model parameters. 
\begin{prop}\label{prop: normnorm}
Assume that $\gamma_Y>\sqrt{2} \gamma_Z$, the unnormalized density $\pi_u$ is given as in \eqref{eq: normal_normal_pi_u} and
let the proposal transition kernel $Q$ be a Gaussian random walk, that is, $Q(x,\cdot)=\mathcal{N}(x,\sigma^2)$
for some $\sigma>0$. Then, there is a Lyapunov function $V\colon G\to [1,\infty)$
for $M_a$, such that $M_a$ is $V$-uniformly ergodic, i.e., Assumption~\ref{ass: metro_V_unif} is satisfied, 
and there are $\beta\in (0,1)$ as well as $k\in\mathbb{N}$ such that
\[
 \norm{i_{2,k}}_{\infty,V^{1-\beta}} <\infty
 \qquad \text{and} \qquad
  \norm{s}_{\infty,V^{\beta}} <\infty.
\]
\end{prop}

The previous proposition implies that there is a constant $D<\infty$, such that
$D_R$ from Corollary~\ref{cor: doub_restr_MCwM} is bounded by $D$ independent of $R$.
Hence there are numbers $\widetilde C_1,\widetilde C_2>0$ such that with $R=\widetilde C_1 \sqrt{N}$ 
and for $N$ sufficiently large we have
\[
  \norm{m_n-m^{(R)}_{n,N}}_{\rm tv} \leq \widetilde C_2\;
  \frac{\log N}{\sqrt{N}}
\]
for any initial distribution $m_0$ on $B_R$.\\[2ex]

\noindent
{\bf Acknowledgments.}
Daniel Rudolf gratefully acknowledges support of the 
Felix-Bernstein-Institute for Mathematical Statistics in the Biosciences 
(Volkswagen Foundation), the Campus laboratory AIMS and the DFG within the project 389483880. Felipe Medina-Aguayo was supported by BBSRC grant BB/N00874X/1 and thanks Richard Everitt for useful discussions.

\begin{appendix}
\section{Technical proofs}

\subsection{Proofs of Section~\ref{sec: quant_est}} 
\label{sec: proof_sec3}

Before we come to the proofs of Section~\ref{sec: quant_est} let us recall a relation between
geometric ergodicity and an ergodicity coefficient. 
Let $V\colon G \to [1,\infty]$ be a measurable, $\pi$-a.e. finite function,
then, define the \emph{ergodicity coefficient $\tau_V(P)$} as
\[
 \tau_V(P) := \sup_{x,y\in G} \frac{\norm{P(x,\cdot)-P(y,\cdot)}_V}{V(x)+V(y)}.
\]
The next lemma provides a relation between the ergodicity coefficient and
$V$-uniform ergodicity.
\begin{lem}  \label{lem: contr_coeff}
 If \eqref{eq: V_unif_erg_ratio} is satisfied, then $\tau_V(P^n)\leq C \alpha^n$.
\end{lem}
A proof of this fact is implicitly contained in \cite{MaZhZh13} and 
can also be found in \cite[Lemma~3.2]{RuSc15}. Both references crucially use an observation 
of Hairer and Mattingly \cite{HaMa11}.

To summarize, 
if the transition kernel $P$ is geometrically ergodic, then, 
by Theorem~\ref{thm: equi} there exist a function $V\colon G \to [1,\infty)$,
$\alpha\in [0,1)$
and $C\in (0,\infty)$ such that, by Lemma~\ref{lem: contr_coeff}, $\tau_V(P^n )\leq C \alpha^n$. 
The next proposition states two further useful 
properties (submultiplicativity and contractivity) of the ergodicity coefficient. 
For a proof of the corresponding inequalities 
see for example \cite[Proposition~2.1]{MaZhZh13}.

\begin{prop} \label{prop: prop_erg_coeff}
  Assume $P,Q$ are transition kernels and $\mu,\nu$ are probability measures on $G$. 
  Then
  \begin{align*}
      \tau_V(P Q) & \leq \tau_V(P)\;\tau_V(Q),\qquad  \text{(submultiplicativity)}\\
   \norm{(\mu-\nu)P}_V & \leq \tau_V(P) \norm{\mu-\nu}_V. \qquad  \text{(contractivity)}
  \end{align*}
\end{prop}

Now we prove Lemma~\ref{lem: aux}.

\begin{proof}[Proof of Lemma~\ref{lem: aux}]
As in the proof of \cite[Theorem~3.1]{Mi05} we use
 \[
  \wt p_n-p_n = (\wt p_0-p_0)P^n + \sum_{i=0}^{n-1} \wt p_i (\wt P-P)P^{n-i-1},
 \]
 which can be shown by induction over $n\in\mathbb{N}$.
 Then
 \begin{equation}  \label{eq: teleskop}
  \norm{\wt p_n  - p_n}_{\rm tv} 
  \leq 
  \norm{(\wt p_0 - p_0)P^n}_{\rm tv}
  + \sum_{i=0}^{n-1} \norm{\wt p_i(\wt P-P)P^{n-i-1}}_{\rm tv}.
  \end{equation}
 With
 Proposition~\ref{prop: prop_erg_coeff} and Lemma~\ref{lem: contr_coeff}
 we estimate the first term of the previous inequality by
 \begin{align*}
  \norm{(\wt p_0 - p_0)P^n}_{\rm tv} 
  \leq  \norm{(\wt p_0 - p_0)P^n}_{V} \leq \tau_V(P^n) \norm{\wt p_0 -p_0}_{V}
  \leq C \alpha^n \norm{\wt p_0 -p_0}_{V}.
 \end{align*}
 For the terms which appear in the sum of \eqref{eq: teleskop} 
 we can use two types of estimates.
 Note that $\tau_1(P) \leq 1$ 
 (here the subscript indicates that $V=1$) which leads 
 by Proposition~\ref{prop: prop_erg_coeff} to 
 \begin{align*}
  &  \norm{\wt p_i(\wt P - P)P^{n-i-1}}_{\rm tv} 
  \leq \norm{\wt p_i(\wt P-P)}_{\rm tv} \tau_1(P^{n-i-1})
  \leq \norm{\wt p_i(\wt P-P)}_{\rm tv} \\
  & = \sup_{\abs{f}\leq 1} \abs{\int_G f(x)\, \wt p_i(\wt P-P)(\dint x)}
    = \sup_{\abs{f}\leq 1} \abs{\int_G (\wt P - P)f(x) \, \wt p_i(\dint x)}\\
  & \leq \int_G \norm{\wt P(x,\cdot)-P(x,\cdot)}_{\rm tv} \, \wt p_i(\dint x) 
  \leq \varepsilon_{{\rm tv},W} \;
  \wt p_i(W).
  \end{align*}
 On the other hand
 \begin{align*}
  &  \norm{\wt p_i(\wt P - P)P^{n-i-1}}_{\rm tv} 
  \leq   \norm{\wt p_i(\wt P - P)P^{n-i-1}}_{V}
  \leq    \norm{\wt p_i(\wt P - P)}_{V}
	  \tau_V(P^{n-i-1})\\
  &\leq	C\alpha^{n-i-1}\norm{\wt p_i(\wt P - P)}_{V}
   \leq C\alpha^{n-i-1} \int_G \norm{\wt P(x,\cdot) - P(x,\cdot)}_{V} \wt p_i(\dint x)\\
  &\leq C\alpha^{n-i-1} \varepsilon_{V,W} \;
  \wt p_i(W).
	   \end{align*}
 Thus, for any $r\in (0,1]$ we obtain
 \begin{align*} 
    \norm{\wt p_i(\wt P - P)P^{n-i-1}}_{\rm tv}
  & \leq \norm{\wt p_i(\wt P - P)P^{n-i-1}}_{\rm tv}^{1-r}
  \cdot \norm{\wt p_i(\wt P - P)P^{n-i-1}}_{\rm tv}^r \\
  & \leq  \varepsilon_{{\rm tv},W}^{1-r}\; \varepsilon_{V,W}^{r} \;
  C^r \; \wt p_i(W)\; \alpha^{(n-i-1)r},
\end{align*}
 which gives by \eqref{eq: teleskop} the final estimate.
\end{proof}

 Next we prove Theorem~\ref{thm: restricted_approx}.

\begin{proof}[Proof Theorem~\ref{thm: restricted_approx}]
 Locally 
 for $x\in B_R$
 we have
 $
  P_R V(x) \leq PV(x)
  \leq \delta V(x) + L,
 $
 and, eventually,
 \begin{align}
    \wt P V(x)  
& \leq P_R V(x) + \abs{\wt P V(x) - P_R V(x)} \notag \\
& \leq \delta V(x) + R \norm{\wt P(x,\cdot) - P_R(x,\cdot)}_{\rm tv} + L \notag \\
& \leq (\delta + R\cdot \Delta(R)) V(x) + L. \label{al: drift_innen}
 \end{align}
 We write $B_R^c$ for $G\setminus B_R$ and obtain for $x\in B_R^c$ that
\begin{equation} \label{eq: drift_ausserhalb}
     \wt P V(x) = \int_{B_R} V(y) \wt P(x,\dint y) \leq V(x).                               
\end{equation}
Denote $\widetilde \delta := \delta + R \cdot \Delta(R) \leq 1/2+\delta/2 <1$. For 
$i\geq 2$ we obtain 
by \eqref{al: drift_innen},
\eqref{eq: drift_ausserhalb} and $(1-\wt \delta^i) \leq 2(1-\wt \delta^{i-1})$ that
\begin{align*}
 \wt p_i(V) & \leq \wt \delta^i \, \int_{B_R} V(x) p_0(\dint x) 
 + (1-\wt \delta^i) \frac{L}{1-\wt \delta} \\
&\qquad + \wt \delta^{i-1} \int_{B_R^c} 
\wt P V(x) p_0(\dint x) + (1-\wt \delta^{i-1}) \frac{L}{1-\wt \delta} \\
& \leq \wt \delta^{i-1} p_0(V) + (1-\wt \delta^{i-1})\frac{3L}{1-\wt \delta}
\leq 6\kappa.
\end{align*}
Furthermore, $p_0(V)\leq \kappa$ and $\wt p_1(V)\leq 2 \kappa$.
Now it is easily seen that
\[
 \sum_{i=0}^{n-1}\wt p_i(V)\alpha^{(n-i-1)r} \leq   \frac{6\kappa}{r(1-\alpha)}.
\]

For $\varepsilon_{{\rm tv},V}$ we have
\[
 \varepsilon_{{\rm tv},V} \leq 
 \max\left\{ \sup_{x\in B_R} \frac{\norm{P(x,\cdot)-\wt P(x,\cdot)}_{\rm tv}}{V(x)},
 \sup_{x\in B_R^c} \frac{\norm{P(x,\cdot)-\wt P(x,\cdot)}_{\rm tv}}{V(x)}
 \right\}.
\]
The second term in the maximum is bounded by $2/R$. 
For $x\in B_R$ we have
\begin{align*}
  \norm{P(x,\cdot)-\wt P(x,\cdot)}_{\rm tv}
 &\leq  \norm{P(x,\cdot)-P_R(x,\cdot)}_{\rm tv}
 + \norm{P_R(x,\cdot)-\wt P(x,\cdot)}_{\rm tv}\\
 &\leq 2P(x,B_R^c)+ \norm{P_R(x,\cdot)-\wt P(x,\cdot)}_{\rm tv}  
\end{align*}
so that the first term in the maximum satisfies
\[
\sup_{x \in B_R} \frac{\norm{P(x,\cdot) - \wt P(x,\cdot)}_{\rm tv} }{V(x)} \leq \Delta(R) 
+ 2 \sup_{x \in B_R} \frac{P(x,B_R^c)}{V(x)}.
\]
Consider a random variable  $X_1^x$ with distribution $P(x,\cdot)$, $x\in B_R$. 
Applying Markov's inequality to the random variable $V(X_1^x)$ leads to
\[
PV(x) =\mathbb{E}[V(X_1^x)] \geq R \cdot \mathbb{P}(V(X_1^x) > R) = R \cdot P(x,B_R^c),
\]
and thus
\[
\sup_{x \in B_R} \frac{P(x,B_R^c)}{V(x)} 
\leq \sup_{x \in B_R} \frac{PV(x)}{R\cdot V(x)}\leq \frac{\delta+L}{R}.
\]
Finally, $ R\cdot \Delta(R)<1-\delta$ and $L\geq 1$ imply
$
 \varepsilon_{{\rm tv},V} \leq \frac{2(L+1)}{R}.
$

We obtain $\varepsilon_{V,V} \leq 2(L+1)$
by the use of
\[
 \norm{P(x,\cdot)-\wt P(x,\cdot)}_{V} 
 \leq PV(x) + \wt PV(x),
\]
the fact that $\sup_{x\in G}\frac{PV(x)}{V(x)} \leq \delta + L$ and
\begin{align*}
\sup_{x\in G}\frac{\wt PV(x)}{V(x)} 
& \leq \max\left\{ \sup_{x\in B_R} \frac{\wt PV(x)}{V(x)}, 
\sup_{x\in B_R^c} \frac{\wt PV(x)}{V(x)} \right\}\\ 
& \underset{ \eqref{al: drift_innen}, \eqref{eq: drift_ausserhalb}}{\leq}
\max\left\{ \wt \delta + L,1 \right\} 
 \leq L+1.
\end{align*}
Then, by Lemma~\ref{lem: aux} for $r\in(0,1]$,
\[
 \norm{p_n-\wt p_n}_{\rm tv} 
 \leq \frac{12 C^{r}(L+1) \kappa}{r\cdot R^{1-r}(1-\alpha)} 
 \leq \frac{12 C(L+1) \kappa}{r\cdot R^{1-r}(1-\alpha)}.
\]
By minimizing over $r$ we obtain for $R\geq \exp(1)$ that
\[
  \norm{p_n-\wt p_n}_{\rm tv} 
   \leq \frac{12 C(L+1) \kappa}{1-\alpha} \cdot \frac{R^{1/\log(R)}\log(R)}{R}.
\]
{Finally by the fact that $R^{1/\log{R}}=\exp(1)<33/12$ the assertion follows}.
\end{proof}

\subsection{Proofs of Section~\ref{sec: approx_MH}} 
\label{sec: proof_Sec4}

We start with the proof of Proposition~\ref{prop: accept_err}.
\begin{proof}[Proof of Proposition~\ref{prop: accept_err}.]
 For any $f\colon G\to \mathbb{R}$ we have
 \begin{align*}
  M_bf(x)-M_cf(x)
  & = \int_G f(y)(b(x,y)-c(x,y))Q(x,\dint y)\\
  & \qquad +f(x)\int_G (c(x,y)-b(x,y)) Q(x,\dint y).
 \end{align*}
In the first case of \eqref{eq: accept_err1case}, we have for all $x\in B$ that
\begin{align*}
 & \norm{M_b(x,\cdot)-M_c(x,\cdot)}_{\rm tv}
   \leq 2 \int_G \abs{b(x,y)-c(x,y)} Q(x,\dint y)\\
 & \leq 2 \int_B b(x,y)\xi(x)(\eta(x)+\eta(y)) Q(x,\dint y)
   \leq 2 \xi(x) (\eta(x) + M_b (\eta \cdot \mathbf{1}_B)(x))\\
 & \leq 2 \xi(x) (\eta(x)+ M_b V^{\beta}(x) \norm{\eta\cdot \mathbf{1}_B}_{\infty,V^\beta} )\\
&\leq 4T \norm{\xi\cdot \mathbf{1}_B}_{\infty,V^{1-\beta}} \norm{\eta\cdot \mathbf{1}_B}_{\infty,V^\beta} V(x),
\end{align*}
where we used that  $\sup_{x\in G}\frac{M_b V(x)}{V(x)} \leq T$ implies  $\sup_{x\in G}\frac{M_b V(x)^\beta}{V(x)^\beta} \leq T^\beta$ by Jensen's inequality. 
Moreover, for any $x\in B$ we obtain
\begin{align*}
 &\norm{M_b(x,\cdot)-M_c(x,\cdot)}_{V}
 \leq \sup_{\abs{f}\leq V} \Big|\int_G f(y) (b(x,y)-c(x,y)) Q(x,\dint y)\\
   &\qquad \qquad 
   + f(x)\left( \int_G (c(x,y)-b(x,y)) Q(x,\dint y) \right) \Big|\\
 &\leq \int_G V(y) \abs{b(x,y)-c(x,y)} Q(x,\dint y)   
 + V(x) \int_G \abs{b(x,y)-c(x,y)} Q(x,\dint y) \\
 &\leq \int_B V(y) b(x,y)\xi(x)(\eta(x)+\eta(y)) Q(x,\dint y) \\  
 & \qquad + V(x) \int_B b(x,y)\xi(x)(\eta(x)+\eta(y)) Q(x,\dint y) \\
 &\leq 2 \norm{\eta\cdot \mathbf{1}_B}_\infty \norm{\xi \cdot \mathbf{1}_B}_\infty (M_bV(x)+V(x)),
\end{align*}
which implies the assertion in that case. 
In the second case of \eqref{eq: accept_err1case}, we have similarly 
for any $x\in B$ that
\begin{align*}
 \quad&  \norm{M_b(x,\cdot)-M_c(x,\cdot)}_{\rm tv}
  \leq 2 \eta(x) M_b (\xi \cdot \mathbf{1}_B) + 2 M_b( \xi\cdot \eta \cdot \mathbf{1}_B)\\
 & \leq 2 \eta(x) \norm{\xi\cdot \mathbf{1}_B}_{\infty, V^{1-\beta}} M_b(V^{1-\beta})(x)
  + 2 \norm{\xi\cdot\eta\cdot \mathbf{1}_B}_{\infty, V} M_b V(x)\\
 & \leq 4 T \norm{\eta\cdot \mathbf{1}_B}_{\infty,V^\beta} \norm{\xi\cdot \mathbf{1}_B}_{\infty,V^{1-\beta}} V(x) 
\end{align*}
and
\begin{align*}
 \norm{M_b(x,\cdot)-M_c(x,\cdot)}_{V}
 &\leq \int_B V(y) b(x,y) \xi(y)(\eta(x)+\eta(y)) Q(x,\dint y)\\   
 &\qquad 
 + V(x) \int_B b(x,y) \xi(y)(\eta(x)+\eta(y)) Q(x,\dint y) \\
 & \leq 2 \norm{\xi\cdot \mathbf{1}_B}_\infty \norm{\eta\cdot \mathbf{1}_B}_\infty (M_bV(x)+V(x)),
\end{align*}
which finishes the proof.
\end{proof}

Before we come to further proofs of Section~\ref{sec: approx_MH} 
we provide some properties of inverse moments of averages of non-negative 
real-valued iid random variables $(S_i)_{i\in \mathbb{N}}$.  
In this setting, the $p$th inverse moment, for $p>0$, is defined by
\[
 j_{p,r} := \left(\mathbb{E}\left( \frac{1}{r} \sum_{i=1}^r S_i \right)^{-p}\right)^{1/p}.
\]
\begin{lem}\label{lemI}
Assume that $j_{p,r}<\infty$ for some $r\in \mathbb{N}$ and $p>0$.
Then
\begin{enumerate}[label=\roman*)]
 \item \label{en: 1_inv_moment}
 $j_{p,s}\leq j_{p,r}$ for $s\in \mathbb{N}$ with $s\geq r$;
 \item \label{en: 2_inv_moment}
  $j_{q,r}\leq j_{p,r}$ for $0<q<p$;
 \item \label{en: 3_inv_moment}
 $j_{k\cdot p, k\cdot r} \leq j_{p,r}$ for any $k\in \mathbb{N}$.
\end{enumerate}

\end{lem}

\begin{proof}
Properties \ref{en: 1_inv_moment} and \ref{en: 2_inv_moment}
follow as in \cite[Lemma~3.5]{MeLeRo15}. For proving \ref{en: 3_inv_moment}  
we have to show that
\[
\E\left[
\left(
\frac{1}{k\cdot r} \sum_{i=1}^{k\cdot r} S_i
\right)^{-p\cdot k}
\right]
\leq 
\E\left[
\left(
\frac{1}{r} \sum_{i=1}^{r} S_i
\right)^{-p}
\right]^k.
\]
To this end, observe first that we can write
\[
\frac{1}{k\cdot r} \sum_{i=1}^{k\cdot r} S_i = \frac{1}{k} \sum_{i=1}^{k} V_i
\]
where the ``batch-means'' $V_1,\dots,V_k$ are non-negative, real-valued iid  random variables
which have the same distribution as $\frac{1}{r} \sum_{i=1}^{r} S_i$. With $Z_i=V_i^{-1}$ 
we obtain
\[
\E\left[
\left(\frac1{
\frac{1}{k\cdot r} \sum_{i=1}^{k\cdot r} S_i}
\right)^{p\cdot k}
\right]
=\E\left[
\left(\frac1{
\frac{1}{k} \sum_{i=1}^{k} \frac1{Z_i}}
\right)^{p\cdot k}
\right]
\]
which is a moment of the harmonic mean of $Z_1,\ldots,Z_k$. 
Using the inequality between geometric and harmonic means 
as well as the independence we find that
\[
\E\left[
\left(\frac1{
\frac{1}{k} \sum_{i=1}^{k} \frac1{Z_i}}
\right)^{p\cdot k}
\right]
\leq \E\left[
\prod_{i=1}^k
Z_i^p
\right]
= \E\left[Z_1^p\right]^k
=\E\left[
\left(\frac1{
\frac{1}{r} \sum_{i=1}^{r} S_i}
\right)^p
\right]^k.
\qedhere
\]
\end{proof}

The previous lemma shows that when inverse moments of 
some positive order are finite, then so are inverse moments of all higher 
and lower orders if the sample size is adjusted accordingly.

\begin{proof}[Proof of Lemma~\ref{lem: doubly_approx_err}]
 It is easily seen that 
 \[
         a(x,z) \mathbb{E}\left[ \min\left\{1,\frac{W_N(x)}{W_N(z)}\right\} \right] \leq a_N(x,z)
 \] for any $x,z\in G$.
 By virtue of Jensen's inequality and $\mathbb{E}[W_N(z)] = 1$ we have $\mathbb{E}[W_N(z)^{-1}]\geq 1$
 as well as 
 \[
  a_N(x,z) \leq \min\left\{ 1, r(x,z) \cdot \mathbb{E}\left[\frac{W_N(x)}{W_N(z)}\right] \right\}
  \leq a(x,z)
 \]
 where we also used the independence of $W_N(x)$ and $W_N(z)$ in the last inequality.
 (The previous arguments are similar to those in \cite[Lemma~3.3 and the proof of Lemma~3.2]{MeLeRo15}.)
 Note that $i_{2,N}(x) \leq i_{2,k}(x)$ for $N\geq k$ by Lemma~\ref{lemI}.
 Hence, one can conclude that
 \begin{align}
&  \abs{a(x,z)- a_N(x,z)}
 \leq a(x,z)
\begin{cases}
 \mathbb{E}\left[\max\left\{0,1-\frac{W_{N}(x)}{W_{N}(z)}\right\}\right] & a(x,z)\geq  a_N(x,z) \\
 \mathbb{E}\left[\frac{W_{N}(x)}{W_{N}(z)}-1\right] & a(x,z)<  a_N(x,z)
\end{cases} \nonumber
\\
 &\qquad \leq 
   a(x,z)\mathbb{E}\abs{1-\frac{W_{N}(x)}{W_{N}(z)}}    
 \leq a(x,z)\, i_{2,N}(z)
      \left( \mathbb{E}\abs{W_{N}(x)-W_{N}(z)}^2 \right)^{1/2} \nonumber\\
    &\qquad \leq a(x,z) i_{2,N}(z) \left[ \left(\mathbb{E}\abs{W_{N}(x)-1}^2\right)^{1/2}	
    + \left(\mathbb{E}\abs{W_{N}(z)-1}^2\right)^{1/2}
    \right] \nonumber \\
		    & \qquad \leq a(x,z) \frac{ i_{2,k}(z)}{\sqrt{N}}(s(x)+s(z)). 
    \qedhere
 \end{align}
\end{proof}

\begin{proof}[Proof of Lemma~\ref{lem: diff_acc_prob1}]
As in the previous proof or 
from \cite[Lemma~3.3 and the proof of Lemma~3.2]{MeLeRo15} an immediate consequence is
 \[
  a(x,z) \mathbb{E}\left[\min\left\{1,\frac{W_{N}(z)}{W_{N}(x)}\right\}\right] \leq
  a_N(x,z) \leq a(x,z)\, \mathbb{E}\left[\frac{W_{N}(z)}{W_{N}(x)}\right].
 \]
 Note that $i_{2,N} \leq i_{2,k}$ for $N\geq k$,
 see Lemma~\ref{lemI}. The rest of the lemma follows as in the previous proof, only the ratio 
 $W_N(x)/W_N(z)$ is reversed.
%
%
\end{proof}
\begin{proof}[Proof of Proposition~\ref{prop: normnorm}]
For random-walk-based Metropolis chains (in particular for $Q$ as assumed in the statement)
by \cite[Theorem~4.1 and the first sentence after the proof of the theorem, as well as, Theorem~4.3, Theorem~4.6]
{JaHa00} we have that $M_a$ is $V_t$-uniformly ergodic with
\begin{align*}
V_{t}(x) & \propto\pi_u(x)^{-t}\propto\exp\left(t\,\frac{\gamma_{Z,Y}}{2}\left(z-x\right)^{2}\right),
\end{align*}
for any $t\in\left(0,1\right)$. 
Hence, Assumption~\ref{ass: metro_V_unif} is satisfied and 
we need to find $t\in(0,1)$ as well as $\beta\in(0,1)$
such that $\left\Vert i_{2,k}\right\Vert _{\infty,V_{t}^{1-\beta}}<\infty$
and $\left\Vert s\right\Vert _{\infty,V_{t}^{\beta}}<\infty$ for some $k\in\mathbb{N}$. 
For showing $\left\Vert s\right\Vert _{\infty,V_{t}^{\beta}}<\infty$
we use \eqref{eq:moment_W_example} to see that
\[
s\left(x\right)\leq \widetilde C
\exp\left( \left(\frac{\gamma_{Z}}{\gamma_{Y}+2\gamma_{Z}}\right)\frac{\gamma_{Z,Y}}{2}\left(z-x\right)^{2}\right),
\]
for some $\widetilde C<\infty$. Hence
\begin{align*}
\frac{s(x)}{V_t(x)^{\beta}} 
& \leq \widetilde C 
\exp\left( \left(\frac{\gamma_{Z}}{\gamma_{Y}+2\gamma_{Z}}-t\beta\right)\frac{\gamma_{Z,Y}}{2}(z-x)^{2}\right) ,
\end{align*}
and choosing $\beta\in(0,1)$ such that 
\begin{equation}
t\beta=\frac{\gamma_{Z}}{\gamma_{Y}+2\gamma_{Z}}\label{eq:NNexamp_condn}
\end{equation}
 leads to $\left\Vert s\right\Vert _{\infty,V_{t}^{\beta}}<\infty$. 
 In order to show $\left\Vert i_{2,k}\right\Vert _{\infty,V_{t}^{1-\beta}}<\infty$,
we first use Lemma~\ref{lemI}~\ref{en: 3_inv_moment} and obtain for any $x\in G$ and any $k \in \mathbb{N}$
\[
i_{2,k}(x)=\mathbb{E}[W_{k}(x)^{-2}]^{\frac{1}{2}}\leq\mathbb{E}\left[W_{1}(x)^{-\frac{2}{k}}\right]^{\frac{k}{2}}.
\]
Then, for $k>2\gamma_{Z}/\gamma_{Y}$ by \eqref{eq:moment_W_example} we have 
\begin{align*}
\mathbb{E}\left[W_{1}(x)^{-\frac{2}{k}}\right]^{\frac{k}{2}} 
 & \propto \exp\left( \left(\frac{\gamma_{Z}\left(1+\frac{2}{k}\right)}{\gamma_{Y}-\frac{2}{k}\gamma_{Z}}\right)
 \frac{\gamma_{Z,Y}}{2}\left(z-x\right)^{2}\right).
\end{align*}
Therefore, there is a constant $\widetilde C<\infty$ such that
\begin{align*}
\frac{i_{2,k}(x)}{V_t(x)^{1-\beta}} & 
\leq \widetilde C
\exp\left(\left(\frac{\gamma_{Z}\left(1+\frac{2}{k}\right)}{\gamma_{Y}-\frac{2}{k}\gamma_{Z}}-t(1-\beta)\right)\frac{\gamma_{Z,Y}}{2}(z-x)^{2}\right).
\end{align*}
We have $\left\Vert i_{2,k}\right\Vert _{\infty,V_{t}^{1-\beta}}<\infty$ if
$\frac{\gamma_{Z}\left(1+\frac{2}{k}\right)}{\gamma_{Y}-\frac{2}{k}\gamma_{Z}} 
 \leq t(1-\beta).$
The latter condition holds whenever
\begin{align*}
k\geq & \frac{2\gamma_{Z}\left(1+t(1-\beta)\right)}{\gamma_{Y}t(1-\beta)-\gamma_{Z}},
\end{align*}
provided that $t(1-\beta)>\gamma_{Z}/\gamma_{Y}$. This implies, by \eqref{eq:NNexamp_condn},
that $t$ should be chosen such that
\begin{align}\label{tcond}
t & >\frac{\gamma_{Z}}{\gamma_{Y}}+\frac{\gamma_{Z}}{\gamma_{Y}+2\gamma_{Z}}.
\end{align}
Choosing $t$ such that it satisfies \eqref{tcond} is feasible whenever the right-hand side of \eqref{tcond} 
is smaller than $1$. This is the case if $\gamma_{Y}>\sqrt{2}\gamma_{Z}$.
\end{proof}

\end{appendix}
\providecommand{\bysame}{\leavevmode\hbox to3em{\hrulefill}\thinspace}
\providecommand{\MR}{\relax\ifhmode\unskip\space\fi MR }
\providecommand{\MRhref}[2]{%
  \href{http://www.ams.org/mathscinet-getitem?mr=#1}{#2}
}
\providecommand{\href}[2]{#2}


\begin{thebibliography}{EJREH17}

\bibitem[ADYC18]{ADYC_2018}
C. Andrieu, A.~Doucet, S.~Y{\i}ld{\i}r{\i}m, and N.~Chopin.
\newblock On the utility of {M}etropolis-{H}astings with asymmetric acceptance
  ratio.
\newblock {\em ArXiv preprint arXiv:1803.09527}, 2018.

\bibitem[AFEB16]{AFEB14}
P.~Alquier, N.~Friel, R.~Everitt, and A.~Boland.
\newblock Noisy {M}onte {C}arlo: Convergence of {M}arkov chains with
  approximate transition kernels.
\newblock {\em Statistics and Computing}, 26(1):29--47, Jan 2016.

\bibitem[AR09]{AnRo09}
C.~Andrieu and G.~Roberts.
\newblock The pseudo-marginal approach for efficient {M}onte {C}arlo
  computations.
\newblock {\em Ann. Statist.}, 37(2):697--725, 2009.

\bibitem[BDH14]{BDH14}
R.~Bardenet, A.~Doucet, and C.~Holmes.
\newblock Towards scaling up {M}arkov chain {M}onte {C}arlo: an adaptive
  subsampling approach.
\newblock In {\em Proceedings of the 31st International Conference on Machine
  Learning}, pages 405--413, 2014.

\bibitem[BRR01]{BrRoRo01}
L.~Breyer, G.~Roberts, and J.~Rosenthal.
\newblock A note on geometric ergodicity and floating-point roundoff error.
\newblock {\em Statist. Probab. Lett.}, 53(2):123--127, 2001.

\bibitem[EJREH17]{EJRE16}
R.~G. Everitt, A.~M. Johansen, E.~Rowing, and M.~Evdemon-Hogan.
\newblock {B}ayesian model comparison with un-normalised likelihoods.
\newblock {\em Statistics and Computing}, 27(2):403--422, Mar 2017.

\bibitem[FHL13]{FHL13}
D.~Ferr{\'e}, L.~Herv{\'e}, and J.~Ledoux.
\newblock Regular perturbation of {$V$}-geometrically ergodic {M}arkov chains.
\newblock {\em J. Appl. Prob.}, 50(1):184--194, 2013.

\bibitem[HM11]{HaMa11}
M.~Hairer and J.~C. Mattingly.
\newblock Yet another look at {H}arris’ ergodic theorem for {M}arkov chains.
\newblock In {\em Seminar on Stochastic Analysis, Random Fields and
  Applications VI}, pages 109--117. Springer, 2011.

\bibitem[JH00]{JaHa00}
S.~Jarner and E.~Hansen.
\newblock Geometric ergodicity of {M}etropolis algorithms.
\newblock {\em Stochastic Process. Appl.}, 85(2):341--361, 2000.

\bibitem[JM17a]{JoMa17A}
J.~E. {Johndrow} and J.~C. {Mattingly}.
\newblock {Coupling and Decoupling to bound an approximating Markov Chain}.
\newblock {\em ArXiv preprint arXiv:1706.02040}, 2017.

\bibitem[JM17b]{JoMa17B}
J.~E. {Johndrow} and J.~C. {Mattingly}.
\newblock {Error bounds for Approximations of Markov chains used in Bayesian
  Sampling}.
\newblock {\em ArXiv preprint arXiv:1711.05382}, 2017.

\bibitem[JMMD15]{Joetal15}
J.~E. {Johndrow}, J.~C. {Mattingly}, S.~{Mukherjee}, and D.~{Dunson}.
\newblock {Optimal approximating Markov chains for Bayesian inference}.
\newblock {\em ArXiv preprint arXiv:1508.03387}, 2015.

\bibitem[MGM06]{MuGhMa06}
I.~Murray, Z.~Ghahramani, and D.~MacKay.
\newblock {MCMC} for doubly-intractable distributions.
\newblock In {\em Proceedings of the 22nd Annual Conference on Uncertainty in
  Artificial Intelligence UAI06}, 2006.

\bibitem[Mit05]{Mi05}
A.~Mitrophanov.
\newblock Sensitivity and convergence of uniformly ergodic {M}arkov chains.
\newblock {\em J. Appl. Prob.}, 42(4):1003--1014, 2005.

\bibitem[MLR16]{MeLeRo15}
F.~J. {Medina-Aguayo}, A.~{Lee}, and G.~{Roberts}.
\newblock {Stability of Noisy Metropolis-Hastings}.
\newblock {\em Stat. Comp.}, 26(6):1187--1211, Nov 2016.

\bibitem[MLR18]{MeLeRo17}
F.~J. {Medina-Aguayo}, A.~Lee, and G.~O. Roberts.
\newblock Erratum to: Stability of noisy {M}etropolis--{H}astings.
\newblock {\em Stat. Comp.}, 28(1):239--239, Jan 2018.

\bibitem[MT96]{MeTw96}
K.~Mengersen and R.~Tweedie.
\newblock Rates of convergence of the {H}astings and {M}etropolis algorithms.
\newblock {\em Ann. Statist.}, 24(1):101--121, 1996.

\bibitem[MT09]{MeTw09}
S.~Meyn and R.~Tweedie.
\newblock {\em {M}arkov {c}hains and {s}tochastic {s}tability}.
\newblock Cambridge University Press, second edition, 2009.

\bibitem[MZZ13]{MaZhZh13}
Y.~Mao, M.~Zhang, and Y.~Zhang.
\newblock {A} generalization of {D}obrushin coefficient.
\newblock {\em Chinese J. Appl. Probab. Statist.}, 29(5):489--494, 2013.

\bibitem[NR17]{NeRo17}
J.~{Negrea} and J.~S. {Rosenthal}.
\newblock {Error Bounds for Approximations of Geometrically Ergodic Markov
  Chains}.
\newblock {\em ArXiv preprint arXiv:1702.07441}, 2017.

\bibitem[PH18]{PaHa18}
Jaewoo Park and Murali Haran.
\newblock Bayesian inference in the presence of intractable normalizing
  functions.
\newblock {\em Journal of the American Statistical Association},
  113(523):1372--1390, 2018.

\bibitem[PS14]{PS14}
N.~Pillai and A.~Smith.
\newblock Ergodicity of approximate {MCMC} chains with applications to large
  data sets.
\newblock {\em ArXiv preprint arXiv:1405.0182}, 2014.

\bibitem[RR97]{RoRo97}
G.~Roberts and J.~Rosenthal.
\newblock Geometric ergodicity and hybrid {M}arkov chains.
\newblock {\em Electron. Comm. Probab.}, 2:no.\ 2, 13--25, 1997.

\bibitem[RRS98]{RoRoSch98}
G.~Roberts, J.~Rosenthal, and P.~Schwartz.
\newblock Convergence pro\-perties of perturbed {M}arkov chains.
\newblock {\em J. Appl. Probab.}, 35(1):1--11, 1998.

\bibitem[RS18a]{RuSc15}
D.~Rudolf and N.~Schweizer.
\newblock {Perturbation theory for Markov chains via Wasserstein distance}.
\newblock {\em Bernoulli}, 24(4A):2610--2639, 2018.

\bibitem[RS18b]{RuSp16}
Daniel Rudolf and Bj{\"o}rn Sprungk.
\newblock On a generalization of the preconditioned crank--nicolson metropolis
  algorithm.
\newblock {\em Foundations of Computational Mathematics}, 18(2):309--343, 2018.

\bibitem[RT96]{RoTw96}
G.~Roberts and R.~Tweedie.
\newblock Geometric convergence and central limit theorems for multidimensional
  {H}astings and {M}etropolis algorithms.
\newblock {\em Biometrika}, 83(1):95--110, 1996.

\bibitem[SS00]{ShSt00}
T.~Shardlow and A.~Stuart.
\newblock A perturbation theory for ergodic {M}arkov chains and application to
  numerical approximations.
\newblock {\em SIAM J. Numer. Analysis}, 37:1120--1137, 2000.

\bibitem[Tie98]{Ti98}
L.~Tierney.
\newblock A note on {M}etropolis-{H}astings kernels for general state spaces.
\newblock {\em Ann. Appl. Probab.}, 8:1--9, 1998.

\bibitem[YR17]{YaRo17}
Jun {Yang} and Jeffrey~S. {Rosenthal}.
\newblock {Complexity Results for MCMC derived from Quantitative Bounds}.
\newblock {\em ArXiv preprint arXiv:1708.00829}, 2017.

\end{thebibliography}
\end{document}